%% file: qs-structures.tex
\documentclass{article}

\usepackage{arxiv}

\usepackage{latexsym}
\usepackage{amsfonts}
\usepackage{amssymb}
\usepackage{amsmath}
\usepackage{color}
\usepackage{textcomp} 
\usepackage{lineno} 
\usepackage{graphicx} 
\usepackage[T1]{fontenc}
\usepackage{mathrsfs}
\usepackage{mathptmx}
\usepackage{helvet}
\usepackage{multirow}
\usepackage{enumitem}

\usepackage{mathrsfs} 
\usepackage{helvet}
\usepackage{courier}
\usepackage{type1cm}
\usepackage{makeidx}
\usepackage{graphicx}
\usepackage{multicol} 
\usepackage{verbatim}
\usepackage{stmaryrd}
\usepackage{fancyhdr}
\usepackage{fancybox} 
\usepackage{tikz}
\usetikzlibrary{patterns,decorations,decorations.pathmorphing,%
arrows,shapes,automata,backgrounds,fit,calc,petri}
\usepackage[fleqn,tbtags]{mathtools}
\usepackage{longtable}
\usepackage{wrapfig} 
\usepackage{algorithmic}
\usepackage{color}
\usepackage{rotating} 
\usepackage{float}
\floatstyle{boxed}
\restylefloat{figure}

\usepackage{hyperref}       

\def\eod{\penalty 1000\hfill\penalty 1000$\diamond$\hskip 0pt\penalty-1000\par}

\newcommand{\qsasTOqscs} {\operatorname{qsas2qscs}}

\newcommand{\drop}[1]{}
\newcommand{\qso}{\textit{qso}}
\newcommand{\qss}{\textit{qss}}
\newcommand{\qsss}{\textit{qsss}}
\newcommand{\QSO}{\textit{QSO}} 
\newcommand{\QSS}{\textit{QSS}} 
\newcommand{\QSSS}{\textit{QSSS}} 
\newcommand{\QSMS}{\textit{QSMS}}
\newcommand{\qsms}{\textit{qsms}} 
\newcommand{\QSAS}{\textit{QSAS}}
\newcommand{\qsas}{\textit{qsas}} 
\newcommand{\QSCS}{\textit{QSCS}}
\newcommand{\qscs}{\textit{qscs}}
\newcommand{\predominants}{\textit{predmt}}
 
\newcommand{\CSCS}{\textit{CSCS}} 
\newcommand{\BCSCS}{\textit{CSCS}^\dagger}
\newcommand{\adding}[3]{[#1 #2 #3]}

\input{macros}

\newtheorem{theorem}{\vspace{3mm}Theorem}[section]
\newtheorem{proposition}[theorem]{\vspace{3mm}Proposition}

\newtheorem{remark}[theorem]{\vspace{3mm}Remark}
\newtheorem{definition}[theorem]{\vspace{3mm}Definition} 
\newtheorem{example}[theorem]{\vspace{3mm}Example}
\newenvironment{proof}{\noindent\textbf{Proof.} \rm}{\hfill $ \Box $}

\title{Quasi-stratified Order Semantics of Concurrency}

\author{ \href{https://orcid.org/0000-0003-4563-1378}{\includegraphics[scale=0.06]{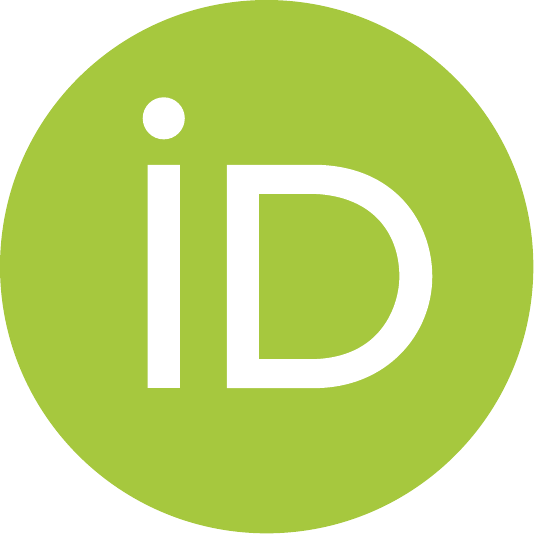}\hspace{1mm}Maciej Koutny}\\
	School of Computing\\
	Newcastle University\\
	1 Science Square, Newcastle upon Tyne, NE4 5TG, U.K. \\
	\texttt{maciej.koutny@ncl.ac.uk} \\
	\And
	\href{https://orcid.org/0000-0002-6711-557X}{\includegraphics[scale=0.06]{orcid.pdf}\hspace{1mm}{\L}ukasz Mikulski} \\
	Faculty of Mathematics and Computer Science\\
	Nicolaus Copernicus University in Toru{\'n}\\
	Chopina 12/18, Toru{\'n}, Poland \\
	\texttt{lukasz.mikulski@mat.umk.pl} \\
}

\begin{document}

\maketitle
 
\begin{abstract}
In the development of operational semantics of concurrent systems, a key decision concerns the 
adoption of a suitable notion of execution model, which basically amounts to choosing a class of partial orders 
according to which events are arranged along the execution line. Typical kinds of such partial orders are
the total, stratified and interval orders. 
In this paper, we introduce quasi-stratified orders --- positioned in-between the stratified and interval orders ---
which are tailored for transaction-like or hierarchical concurrent executions.

Dealing directly with the vast number of executions of concurrent system is far from being practical.
It was realised long time ago that it can be much more effective to consider behaviours at a more abstract level of 
behavioural specifications (often based on intrinsic relationships
between events such as those represented by causal partial orders),
each such specification --- typically, a relational structure --- encompassing a (large) number
of executions. 

In this paper, we introduce and investigate suitable specifications for behaviours represented by quasi-stratified orders.
The proposed model of quasi-stratified relational structures is based on two 
relationships between events --- the `before' and ‘not later’ relationships --- which can be used to 
express and analyse causality, independence, and simultaneity between events. 
\end{abstract}


\section{Introduction}
 
When designing a sound semantics of a class of concurrent systems, 
it is paramount to address the following two aspects: (i) 
faithful formal modelling of the executed behaviours; and (ii) 
providing support for effective validation techniques of 
system behaviour.
The former basically amounts to choosing a class of partial orders 
(\eg total, stratified or interval orders)
according to which events are arranged along the execution line. 
The latter is often realised through devising suitable relational structures which 
are more abstract and encompassing a (large) number
of executed behaviours.
Typically, such relational structures capture intrinsic relationships
between events (such as those represented by causal partial orders),
and can be used to express and analyse
causality, independence, and simultaneity between events. 

Each of the existing ways of modelling 
concurrent behaviours using partial orders comes with its own merits.
We feel, however, that neither addresses in a satisfactory manner the behavioural features 
of systems which support hierarchical executions, such as 
those built on transaction-like executions.
In this paper, we aim at providing some initial ideas and solutions to 
rectify this.

\subsection{Background}

Virtually any approach to modelling operational semantics of concurrent systems 
is fundamentally based on a model for representing individual system behaviours.
In such representations, two events (executed actions) 
can be observed as either happening one after another or simultaneously.
This, in particular, means that such a behaviour can be represented as 
a partially ordered set of events as precedence is a transitive relationship, and the
simultaneity between events can be understood as the relationship of being unordered. 
In the literature, one can identify three main kinds of 
partial orders modelling concurrent behaviours:
(i) total orders modelling sequential executions;
(ii) stratified orders modelling behaviours in which sets of simultaneous events are executed sequentially;
and
(iii) interval orders modelling behaviours where events can take time and the corresponding time intervals can 
arbitrarily overlap. 
In this paper we are concerned with modelling of executions where events can have duration,
and so can be modelled by interval orders.

The relevance of interval orders follows from an observation, credited to Wiener~\cite{Wie},
that any execution of a physical system that can be observed by a single observer
must be an interval order. It implies that the most precise observational
semantics should be defined in terms of interval orders (cf.~\cite{JanKou93}).
In the area of concurrency theory,
the use of interval orders 
can be traced back to~\cite{Lam86,DBLP:journals/ijpp/Pratt86,JanKou93}, and
processes of concurrent systems with interval order semantics were studied in~\cite{JanKou97,DBLP:journals/iandc/JanickiY17}.
Interval orders were used to investigate communication protocols in~\cite{ABM}, using the approach of~\cite{Lam78}.
Interval semantics (ST-semantics) was investigated for Petri nets with read arcs~\cite{Vog02} and discussions on distributability of concurrent systems~\cite{Glabb}.

As an example, consider four transactions in the distributed environment: 
$a$, $b$, $c$, and $d$. Moreover, suppose that $a$ precedes $b$ and $c$
precedes $d$. Suppose further that $a$ does not precede $d$ and $c$ does not precede $b$.
Then, it is possible for two messages $\alpha$ (from $d$ to $a$) and 
$\beta$ (from $b$ to $c$) to be sent and delivered. Hence 
$\alpha_\textit{snd}\prec\alpha_\textit{rcv}\prec
\beta_\textit{snd}\prec\beta_\textit{rcv}\prec\alpha_\textit{snd}$, which is impossible.
Hence $a$ precedes $d$ or $c$ precede $b$, and so the precedence 
relationship between the four transactions is an interval order. 

In this paper, we consider a model of interval order executions tailored for concurrent systems 
such as distributed systems with nested transaction, and concurrent systems employing action
refinement. In such cases, the underlying assumption is that 
an execution can be seen as a sequence of generalised strata,
where each generalised stratum is composed, \eg of  a base event $x$ which is simultaneous 
with two events, $y$ and $z$, and $y\prec z$. Executions of this type 
can be represented by interval orders of special type, which we call 
\emph{quasi-stratified}. 

\subsection{The approach followed in this paper}

Dealing directly with the vast number of actual behaviours of a concurrent system is far from being practical.
It was realised long time ago that it is more effective to discuss behaviours at more abstract level of 
behavioural specifications (often based on intrinsic relationships
between events such as those represented by causal partial orders),
each such specification --- typically, a relational structure --- encompassing a (large) number
of behaviours.

 During the operation of a 
concurrent or distributed system, system actions
may be executed sequentially or simultaneously. To faithfully reflect the simultaneity of actions, one 
may choose step sequences (or stratified orders of events) where 
actions are regarded as happening instantaneously in sets (called steps), 
or interval orders where actions are executed over time intervals. 
In this paper, we investigate a novel representation of concurrent behaviours
(such as those exhibited by transactional distributed systems),
called quasi-stratified orders, which lie in-between the stratified and interval orders.
 
Dealing with the semantics of concurrent systems purely in terms of their individual executions, such as total orders (sequences), stratified orders (step sequences), 
quasi-stratified orders,
or interval orders is far from being computationally efficient in terms both of behaviour modelling and of 
property validation. 
To address this shortcoming, more involved relational structures have been introduced, aiming at a succinct and faithful
representation $\rs$ of (often exponentially large) sets $\rss$ of closely related individual executions. 
Examples include causal partial orders for sequences, and 
invariant structures for step sequences. 
Succinctness is usually achieved by 
retaining in $\rs$ (through intersection) only those relationships which are common to all executions in $\rss$.
Faithfulness, on the other hand, requires that all potential executions which are extensions of $\rs$ belong to $\rss$.
Structures like $\rs$, 
referred to as invariant structures 
represent concurrent histories 
and both desired properties 
(succinctness and faithfulness) 
follow from generalisations of Szpilrajn's Theorem (\ie a partial order is the intersection of all
its total order extensions).
Although $\rs$ provides a clean theoretical capture of the set $\rss$, to turn them into a practical tool (as, \eg in~\cite{DBLP:conf/cav/McMillan92}) 
one needs to be able to derive them directly from single executions using 
the relevant structural properties of the concurrent system. 
This brings into focus relational structures with acyclic relations on events (\eg dependence graphs introduced
in~\cite{DBLP:conf/ac/Mazurkiewicz86} and analysed in detail in~\cite{HoogRoz95}), which 
yield invariant structures after applying a suitable closure operation (\eg the transitive closure 
for acyclic relations). 

The approach sketched above has been introduced and investigated 
in~\cite{DBLP:series/sci/2022-1020,DBLP:journals/tcs/JanickiKKM21}
as a generic model 
that provides general recipes for building analytic frameworks based on 
relational structures.
It starts from 
a class $\RR$ of relational structures (such as acyclic relations), 
which are compared w.r.t.\ the information they convey 
(expressed by the relation $\pref$ with $\rs'\pref\rs$ if
$\rs$ is obtained from $\rs'$ by adding
new relationships, \ie $\rs$ is more \emph{concrete} than $\rs'$). 
Then the \emph{maximal} relational structures 
(such as total orders) 
$\RR^\sat \subseteq \RR$ represent individual executions. 
Then \emph{closed} 
relational structures (like causal partial orders) $\RR^\closed \subseteq \RR$
can be used to analyse intrinsic relationships between executed events. 
 
\subsection{Contribution of this paper}

There is a vast volume of literature concerned with the semantics 
in which the three established ways 
 of modelling executed behaviours (\ie those based on total, stratified and interval orders)
are used. This paper paves the way for a new kind of semantics 
based on quasi-stratified orders. 
In particular, we introduce for the first time (as far as we are aware) the notions
of quasi-stratified order and the corresponding specification technique based
on quasi-stratified relational structures.

Throughout the paper we follow closely the general approach 
to defining order-theoretic semantics of concurrent systems
formalised 
in~\cite{DBLP:journals/tcs/JanickiKKM21,DBLP:series/sci/2022-1020},
taking advantage of the existing result where it is possible. 
It captures behaviour specifications (concurrent histories) 
which are based on two relationships between events, namely \emph{precedence}
and \emph{weak precedence}, intuitively corresponding to the `earlier than' and 
`not later than' relative positions in individual executions. 
The same two relationships were discussed in~\cite{JanKou93,DBLP:series/sci/2022-1020,JKKM-PN2024}, and the resulting 
relational structures proved to be expressive 
enough to faithfully model inhibitor nets, Boolean networks, reaction systems, membrane systems, 
etc.~\cite{DBLP:series/sci/2022-1020,KleKou06}.

\subsection{Structure of this paper}
 
In the next section, we recall basic facts involved in the modelling of 
concurrent behaviours using the total, stratified, and interval orders. 
Moreover, we sketch the way in which relational structures with two
relationships have been designed to serve as specifications of sets of related 
behaviours. 
Sections~\ref{sect-qso} and~\ref{sect-qscs} introduces quasi-stratified orders and quasi-stratified acyclic structures, and 
makes precise their algebraic and operational interpretations. 
Section~\ref{sec-max} introduces maximal quasi-stratified structures and
shows that the quasi-stratified structures provide a suitable specification
model for quasi-order behaviours. 
Section~\ref{sec-inv} is concerned with closed quasi-stratified order structures and 
the corresponding structure closure mapping. 

\section{Modelling concurrent behaviours}
\label{sect-prelim}
 
There is a natural distinction between specifications of concurrent 
system behavior (which may be given as sets of constraints about 
the ordering of actions to be executed), and actual observed behaviours
(executions of actions) which should adhere to these specifications.
Usually, a single specification can be realised by (very) many actual behaviours.

In this paper, we will address both specifications and behaviours within a single framework.
We will consider events (\ie executed actions) 
as activities taking time which can therefore be understood as 
time intervals. We are not interested in the precise timings of such intervals,
but only whether they are ordered in time 
(\ie that one interval precedes another interval) or simultaneous 
(\ie that two time intervals overlap). 

\subsection{Partial order models of system executions}
\label{subs-1}

We adopt the view that observed behaviours of concurrent systems 
are represented by partial orders (as acyclicity of events
is a very minimal requirement resulting from physical considerations, and event precedence 
is transitive).
Note that for partial order system executions, simultaneity 
of events
(corresponding to the overlapping of time intervals)
amounts to the lack of ordering between events.

There are different partial order models of concurrent 
behaviours, \eg reflecting the 
underlying hardware and/of software.
In the literature, there are three main kinds of such models,
namely the total, stratified, and interval orders, as recalled below. 
 
Let $\Delta$ be a finite set (of events) and $\prec$ be
a binary (precedence) relation over $\Delta$.
Then $\po=\FA{\Delta,\prec}$ is called (below $x,y,z,w\in\Delta$):
\begin{itemize} 
\item
\emph{partial order} if 
\\
$ 
\begin{array}{l@{~~}l@{~~}l@{~~~~~~~}l@{~~}l@{~~}l}
 \textsc{po{:}1}
 &: 
 &\makebox[3cm][l]{\mbox{$x\not\prec x$}} 
 &
 \textsc{po{:}2} 
 & :
 & x\prec y \prec z \Longrightarrow x\prec z
\end{array}
$ 

\item 
\emph{total order} if 
\\
$ 
\begin{array}{l@{~~}l@{~~}l@{~~~~~~~}l@{~~}l@{~~}l}
 \textsc{to{:}1}
 &:
 & \makebox[3cm][l]{\mbox{$x\not\prec x$}} 
 & 
 \textsc{to{:}3} 
 & :
 & x\neq y \Longrightarrow x\prec y~\vee~y\prec x 
\\ 
 \textsc{to{:}2} 
 & :
 & x\prec y \prec z \Longrightarrow x\prec z 
\end{array}
$

\item
\emph{stratified order} if 
\\
$
\begin{array}{l@{~~}l@{~~}l@{~~~~~~~}l@{~~}l@{~~}l}
 \textsc{so{:}1}
 &:
 & \makebox[3cm][l]{\mbox{$x\not\prec x$}} 
 & \textsc{so{:}3} 
 & :
 & x\not\prec y\not\prec x ~\wedge~ x\prec z \Longrightarrow y\prec z 
\\ 
\textsc{so{:}2} 
 & :
 & x\prec y \prec z \Longrightarrow x\prec z
 &\textsc{so{:}4} 
 & :
 & x\not\prec y\not\prec x ~\wedge~ z\prec x \Longrightarrow z\prec y
\end{array}
$
 
\item 
\emph{interval order} if 
\\
$
\begin{array}{l@{~~}l@{~~}l@{~~~~~~~}l@{~~}l@{~~}l}
 \textsc{io{:}1}
 &:
 &\makebox[3cm][l]{\mbox{$x\not\prec x$}} 
 & 
 \textsc{io{:}2} 
 & :
 &
 x\prec y~\wedge~z\prec w \Longrightarrow x\prec w~\vee~z\prec y 
\end{array}
$
\end{itemize} 
All total orders are stratified, all stratified orders are interval, 
and all interval orders are partial.

The following result 
provides a useful characterisation of stratified orders.

\begin{proposition} 
\label{prec-uuggygyg}
$ \FA{\Delta,\prec}$ is a stratified order iff 
there is a partition $\Delta_1,\dots,\Delta_n$ of $\Delta$ such that 
$\prec=\bigcup_{1\leq i<j\leq n}\Delta_i\times\Delta_j$.
\end{proposition}

\begin{proof}
($\Longleftarrow$) 
Clearly, \textsc{so{:}1,2} are satisfied.
Moreover, \textsc{so{:}3,4} hold
as $x\not\prec y\not\prec x$ implies that
there is $\Delta_i$ such that $x,y\in\Delta_i$.

($\Longrightarrow$)
We first observe that the relation
$S=\{\FA{x,y}\in\Delta\times\Delta\mid x\not\prec y\not\prec x\} $
is an equivalence relation.
Indeed, $\{\FA{x,x}\mid x\in\Delta\}\subseteq S$ by \textsc{so{:}1}.
Suppose now that
$x\not\prec y\not\prec x$ and $z\not\prec y\not\prec z$.
If $x\prec z$ or $z\prec x$, then, by \textsc{so{:}3,4}, we obtain that
$y\prec z$ or $z\prec y$, yielding
a contradiction with $z\not\prec y\not\prec z$. 
Hence, there is a partition $\Delta_1,\dots,\Delta_n$ of $\Delta$ into the equivalence classes of $S$.

Suppose that $i\neq j$, $x,y\in \Delta_i$ and $z\in\Delta_j$. Then, as $\Delta_i$ and $\Delta_j$ 
are equivalence classes of $S$,
$x\prec z$ or $z\prec x$. If $x\prec z$, then, by \textsc{so{:}3}, $y\prec z$.
Hence, $\Delta_i\times \Delta_j\subseteq \prec$. Moreover, as $\prec$ is acyclic (by \textsc{so{:}1,2}),
we have
$(\Delta_j\times \Delta_i)\cap\prec=\es$. Similarly, if $z\prec x$, then 
$\Delta_j\times \Delta_i\subseteq \prec$ and $(\Delta_i\times \Delta_j)\cap\prec=\es$.

As a result, for all different $\Delta_i$ and $\Delta_j$, 
$\prec|_{(\Delta_i\cup\Delta_j)\times(\Delta_j\cup\Delta_i)}$
is equal either to $\Delta_i\times \Delta_j$ or $\Delta_j\times \Delta_i$.
Therefore, because $\prec$ is acyclic, there is a re-ordering of the 
equivalence classes of $S$ such that 
$\prec=\bigcup_{1\leq i<j\leq n}\Delta_i\times\Delta_j$. 
\end{proof}

\begin{example}
\label{ex:orders}
\input{figures/orders.tex}
\end{example}

\begin{remark}[models of execution]
 \label{rem-1}
 
Thinking about partial orders as representations of behaviours, 
leads to the following observations.
\begin{itemize}
\item 
The adjective in `total order' derives from an alternative 
definition: 
\begin{quote}
$\po=\FA{\Delta,\prec}$ is a total order  
if $\Delta$ can be enumerated as $x_1,\dots,x_n$ so that 
$\prec=\{\FA{x_i,x_j}\mid 1\leq i<j\leq n\}$.
\end{quote} 
Intuitively, this means that the intervals corresponding to events are linearly 
ordered, with no overlapping (no two events are simultaneous).

\item
The adjective in `stratified order' derives from an alternative 
definition (see Proposition~\ref{prec-uuggygyg}): 
\begin{quote}
$\po=\FA{\Delta,\prec}$ is a stratified order
if there are sets (strata)
$\Delta_1,\dots,\Delta_n$ forming a partition of $\Delta$ so that 
$\prec=\bigcup_{1\leq i<j\leq n}\Delta_i\times\Delta_j$.
\end{quote} 
Intuitively, this means that the intervals corresponding to events can be 
partitioned into non-overlapping totally ordered strata (or steps) such that within each stratum 
all the intervals overlap.

\item
The adjective in `interval order' derives from an alternative 
definition (see~\cite{Fish}):
\begin{quote}
$\po=\FA{\Delta,\prec}$ is an interval order
if there exist real-valued mappings, 
$\beta$ and $\epsilon$, such that, for
all $x,y\in\Delta$, $\beta(x)< \epsilon(x)$ and 
$x\prec y \iff \epsilon(x)< \beta(y)$.
\end{quote}
Intuitively, $\beta$ and $\epsilon$ represent 
the `beginnings' and `endings', respectively, of the intervals corresponding to 
events in $\Delta$. 
A remark, credited to Wiener~\cite{Wie}, states
that any execution of a physical system that can 
be observed by a single observer is an interval order.
Interval order executions are typically generated 
after splitting each action in the model of a 
concurrent system into a `begin' and an `end' actions 
and, after executing the modified model
using sequential semantics, deriving 
the corresponding interval orders.
However, this can also be done directly, \ie without 
splitting actions, as shown in~\cite{MKMPK03}.
\eod
\end{itemize} 
\end{remark}

In this paper, we will introduce a new partial order model of concurrent behaviours,
called quasi-stratified orders, which
lies in-between the stratified and interval orders (in the sense presented in Examples~\ref{ex:orders} and~\ref{ex-qsorder}). 
The motivation comes from 
a realisation that in many systems (such as distributed transaction systems) a base activity 
can exist so that other (sub-activities) are executed and completed over its lifetime.
 
\subsection{Relational structures for specifying system executions}
\label{subs-2}

Concurrent system behaviours can be specified using relational structures.
Here we consider relational structures with two relationships `generating' behaviours 
consistent with these two relations. 

Throughout the paper, a \emph{structure} is a triple $\rs=\FA{\Delta,\prec,\sq}$,
where $\prec$ and $\sq$
are two binary relations over a finite domain $\Delta$.
Moreover, if both relations are irreflexive, 
then $\rs$ is a \emph{relational structure}. 
We also use 
$\Delta_\rs$, $\prec_\rs$, and $\sq_\rs$ to denote the components of $\rs$.
The notations below can be used to add new domain elements and relationships: c
\[
\rs[x]=\FA{\Delta\cup\{x\},\prec,\sq} 
~~~~~~~~ 
\rs{\adding{x}{\prec}{y}}=\FA{\Delta,\prec\cup\{\FA{x,y}\},\sq}
~~~~~~~~
\rs{\adding{x}{\sq}{y}}=\FA{\Delta,\prec,\sq\cup\{\FA{x,y}\}} 
\]
where in the first case $x\notin\Delta$ while in the other two cases $x,y\in\Delta$.

In the approach followed in this paper, relational structures are
used to specify some relevant
relationships between events represented by domain elements.
Intuitively, $x\prec y$ can be interpreted as stating that the execution of $x$ precedes the 
execution of $y$, whereas $x\sq y$ can be interpreted as stating that the execution of $x$ precedes or is simultaneous with the 
execution of $y$ (\ie $y$ does not precede $x$)\footnote{It is never the case that $x$ precedes $y$ while $y$ weakly precedes $x$. Hence, $\sq$ can be treated as the negation of the reversed $\prec$.}.

Let $\RR$ be a set of relational structures, and $\rs,\rs'$ be relational structures: 
\begin{itemize} 
\item 
$\rs \pref\rs'$ means that $\rs'$ is an extension of $\rs$, 
\ie
$\Delta_\rs= \Delta_{\rs'}$,
$\prec_\rs\subseteq \prec_{\rs'}$, and
$\sq_\rs\subseteq \sq_{\rs'}$.
 
\item 
 $\Ext{\RR}{\rs}=
 \{\rs'\in\RR\mid\rs \pref\rs'\}$
are the extensions of $\rs$ in $\RR$.
 
\item 
 $\RR^\sat=\{\rs'\in\RR\mid
 \Ext{\RR}{\rs'}=\{\rs'\}\}$ are the maximal structures in $\RR$.
 
\item
 $\satmap_\RR(\rs)=\Ext{\RR^\sat}{\rs}$ are 
 the maximal extensions (or saturations) of $\rs$ in $\RR$.
 
\item 
 $\rs|_\Phi=\FA{\Phi,\prec|_{\Phi\times \Phi}, \sq|_{\Phi\times \Phi}}$
 is the projection of $\rs$ onto $\Phi\subseteq \Delta_\rs$.

\item
 $\RR$ is projection-closed if $\rs|_\Phi\in\RR$, 
 for all $\rs\in\RR$ and $\Phi\subseteq\Delta_\rs$.
 \item
 $\RR$ is intersection-closed if
 $\FA{\Delta_\rs\cap\Delta_{\rs'},\prec_\rs\cap\prec_{\rs'},\sq_\rs\cap\sq_{\rs'}}\in \RR$,
 for all $\rs,\rs'\in\RR$.
\end{itemize}

A partial order $\po=\FA{\Delta,\prec}$ can be represented 
in a natural way as a relational structure 
\[\rho(\po) =\FA{\Delta,\prec,\{\FA{x,y}\in\Delta\times\Delta\mid y\not\prec x\neq y\}},\]
which means that $\prec$ is included in $\sq$ and
all unordered events are simultaneous.

As $\RR^\sat$ are intended to represent executions specified by relational
structures in $\RR$, we immediately face 
a key modelling question:
\begin{quote}
 \emph{What relational structures $\RR$ provide a 
 suitable specification model for partial orders $\mathcal X$ ?}
\end{quote}
The answer is rather straightforward, namely all one needs to require
is that $\RR^\sat=\rho(\mathcal X)$.
Note that in such a case, $\rs\in\RR$ iff 
$\satmap_\RR(\rs)\cap\rho(\mathcal X)\neq\es$. The latter 
implies that no specification is vacuous as it generates
at least one execution.

\begin{remark}[acyclity in execution]
 \label{rem-2}

 For the partial order models of behaviours in Section~\ref{subs-1}, 
 suitable relational structures $\rs=\FA{\Delta,\prec,\sq}$ satisfy 
the following:
\begin{itemize}
\item 
Total orders: 
there are no events $x_1,\dots,x_n(=x_1)$
such that $x_i(\prec\cup\sq) x_{i+1}$,
for every $1\leq i< n$. 

\item 
Stratified orders: 
there are no events $x_1,\dots,x_n(=x_1)$
such that: (i) $x_i(\prec\cup\sq) x_{i+1}$,
for every $1\leq i< n$; and (ii) 
there is $i< n$ such that $x_i \prec x_{i+1}$. 

\item
Interval orders: 
there are no events $x_1,\dots,x_n(=x_1)$
such that: (i) $x_i(\prec\cup\sq) x_{i+1}$,
for every $1\leq i< n$; 
and (ii) 
$x_{i-1}\prec x_i~\vee~x_i\prec x_{i+1}$, for all $1<i\leq n$,
and $x_{n-1}\prec x_n~\vee~x_1\prec x_2$.
\end{itemize}
Each case above captures some kind of `acyclicity',
or perhaps more accurately, `constrained cyclicity'.
Such a property is linked to the absence of substructures of $\rs$
which are deemed `unacceptable' or `forbidden' in a given execution model. For example,
if $x\sq y\sq x$ then no total order exists in which $x$ does not precede $y$
and, at the same time, $y$ does not precede $x$. However, such a cyclic relationship
is allowed in the stratified order case as one can execute $x$ and $y$ simultaneously. 
\eod
\end{remark}

\begin{example}
\label{ex:fcycles}
\input{figures/fcycles.tex}
\end{example}


\section{Quasi-stratified orders}
\label{sect-qso}

As far as we are aware
the next definition introduces a new class of partial 
orders, positioned in-between stratified and
interval orders.

\begin{definition}[\textsc{qs}-order]
\label{def-qso} 
\emph{Quasi-stratified orders} (or \textsc{qs}-orders) 
are the smallest set $\QSO$ such that $\FA{\es,\es}\in\QSO$ and the following hold, 
for all $\qso=\FA{\Delta,\prec}$ and $\qso'=\FA{\Delta',\prec'}$ in $\QSO$:
\begin{itemize} 
\item 
 If $x\notin \Delta$,
 then $ \qso[x]=\FA{\Delta\cup\{x\},\prec}$ belongs to $\QSO$.
\item 
 If $\Delta\cap\Delta'=\es$, then
$\qso\circ\qso'=\FA{\Delta\cup \Delta', 
(\prec\cup\prec')\cup (\Delta\times \Delta')}$ belongs to $\QSO$.
\end{itemize}
Moreover, $\qso\in \QSO$ is a \textsc{qso}-stratum
if there are $\qso'\in\QSO$ and $x$ such that $\qso=\qso'[x]$.
\eod
\end{definition} 
 
\begin{proposition}
\label{prop-uqiqef}
Let $\qso=\FA{\Delta,\prec}\in\QSO$. Then:
\begin{enumerate}
 \item 
 $\qso\circ\FA{\es,\es}=\FA{\es,\es}\circ\qso=\qso$.
 \item
 $\qso|_\Phi=\FA{\Phi,\prec|_{\Phi\times\Phi}}\in\QSO$, for every $\Phi\subseteq \Delta$; 
 \ie \textsc{qs}-orders are projection-closed.  
 
\item 
$\qso$ is a \textsc{qso}-stratum iff 
$\qso\neq\FA{\es,\es}$ and there is 
$x\in\Delta_\qso$ such that 
$x\not\prec_\qso z\not\prec_\qso x$, for every $z\in\Delta_\qso\setminus\{x\}$.

\item 
If $\qso\neq\FA{\es,\es}$, then there is a unique
sequence $\qso_1,\dots,\qso_n$ ($n\geq 1$)
of \textsc{qso}-strata 
such that $\qso=\qso_1\circ\dots\circ\qso_n$.
\end{enumerate} 
\end{proposition}
\begin{proof}
All parts follows directly from Definition~\ref{def-qso}. We only note that
`$\circ$' is an associative operation, and
$ 
 (\qso'\circ\qso'')|_\Psi= \qso'|_{\Psi\cap\Delta_{\qso'}}\circ\qso''|_{\Psi\cap\Delta_{\qso''}}$, 
 for all $\qso',\qso''\in\QSO$ with disjoint domains and
 $\Psi\subseteq \Delta_{\qso'}\cup \Delta_{\qso''}$.
\end{proof}
 
The following result provides an axiomatic characterisation of \textsc{qs}-orders, 
similar to those given for stratified and interval orders.

\begin{theorem}
\label{prop-axioms_qso}
Let $\rs=\FA{\Delta, \prec}$, where $\prec$ is a binary relation over 
a finite set $\Delta$.
Then $\rs\in\QSO$ iff the following hold, for all $x,y,z,t\in \Delta$:
\begin{equation*}
\label{axioms-qso}
\begin{array}{lll@{~~~}c@{~~~}l@{~~~~~}l@{~~~}lr@{~~~}c@{~~~}l@{~~~~~}l@{~~~}l}
 \textsc{qso{:}1}
 &:
 & 
 x\not\prec x 
 \\ 
 \textsc{qso{:}2} 
 & :
 &
 x\prec y \land z \prec t
 & \Longrightarrow
 & x\prec t \land z\prec y \,~\lor~\,x\prec z \land x\prec t \,~\lor~\,
 z\prec x \land z\prec y \,~\lor~\, 
 t\prec y \land z\prec y \,~\lor~\, 
 y\prec t \land x\prec t 
\end{array}
\end{equation*} 
\end{theorem}
\begin{proof}
($\Longrightarrow$) 
We proceed by induction on the derivation of $\rs\in\QSO$ (see Definition~\ref{def-qso}).

First, we observe that $\rs=\FA{\es,\es}$ trivially satisfies \textsc{qso{:}1,2}.
In the inductive case, we consider the following two cases.

\emph{Case 1:}
$\rs=\qso[x]$, where $\qso\in\QSO$ and $x\notin\Delta$. Then, by the induction 
hypothesis, $\qso$ satisfies \textsc{qso{:}1,2}. Hence,
$\rs$ also satisfies \textsc{qso{:}1,2} as we neither introduced nor 
removed any relationships.

\emph{Case 2:}
$\rs=\qso'\circ\qso''$, where $\qso'=\FA{\Delta',\prec'}\in\QSO$ and 
$\qso''=\FA{\Delta'',\prec''}\in\QSO$ satisfy 
$\Delta'\cap\Delta''=\es$.
Then we may additionally assume that $\Delta'\neq\es\neq\Delta''$ as $\qso'=\FA{\es,\es}$ 
(or $\qso''=\FA{\es,\es}$) implies $\rs=\qso''$ (resp.\ $\rs=\qso'$). As a result, by the 
induction hypothesis, both $\qso'$ and $\qso''$ satisfy \textsc{qso{:}1,2}.
Hence, $\rs$ satisfies \textsc{qso{:}1} as we introduced new relationships only between disjoint sets $\Delta'$ and $\Delta''$. 
\\
To show that \textsc{qso{:}2} is satisfied by $\rs$, we take any $x,y,z,t\in\Delta'\cup\Delta''$ such that $x\prec y$ and $z\prec t$. 
As 
$y\in\Delta'$ / $t\in\Delta'$ / $x\in\Delta''$ / $z\in\Delta''$
implies
$x\in\Delta'$ / $z\in\Delta'$ / $y\in\Delta''$ / $t\in\Delta''$, respectively,
we have the following cases to consider (after omitting the symmetric ones):
\begin{itemize}
\item 
 $x,y,z,t\in\Delta'$ or $x,y,z,t\in\Delta''$. Then
 \textsc{qso{:}2} is satisfied as
 it holds for both $\qso'$ and $\qso''$.
 
\item 
 $x\in\Delta'$ and $y,z,t\in\Delta''$. Then we have $x\prec z$ and $x\prec t$.
 
\item 
 $x,y,z\in\Delta'$ and $t\in\Delta''$. Then 
 we have $z\prec y$ and $t\prec y$.
 
\item 
 $x,y\in\Delta'$ and $z,t\in\Delta''$. Then 
 we have $x\prec z$ and $x\prec t$.
 
\item 
 $x,z\in\Delta'$ and $y,t\in\Delta''$. Then 
 $x\prec t$ and $z\prec y$.
\end{itemize}
Hence $\rs$ satisfies \textsc{qso{:}2}. 

($\Longleftarrow$) 
Suppose that $\rs=\FA{\Delta,\prec}$ satisfies \textsc{qso{:}1,2}.
We first observe that $\rs$ is an interval order (and so also a partial order). 
Indeed, since \textsc{qso{:}1}
holds for $\rs$, 
all we need to show is that $x\prec y \wedge z\prec t$ implies $x\prec t \wedge z\prec y$. 
And this, clearly, follows from \textsc{qso{:}2}.

We then prove 
the ($\Longleftarrow$) implication by induction on $|\Delta|$.
To start with, if $|\Delta|\leq 1$ then, by \textsc{qso{:}1}, $\rs\in\QSO$.
In the inductive case $|\Delta|>1$ and we consider the following two cases.

\emph{Case 1:}
There is $x\in\Delta$ such that $x\not\prec y\not\prec x$,
 for every $y\in\Delta$. 
 Then, by the induction hypothesis, 
 $\rs' = \FA{\Delta\setminus\{x\},\prec}$ is an \textsc{qs}-order.
 Moreover, $\rs = \rs'[x]$ and so,
 by Definition~\ref{def-qso}, $\rs$ also is an \textsc{qs}-order.
 
\emph{Case 2:}
 For every $x\in\Delta$ there is $y\in\Delta$ such that $x\prec y$ or $y\prec x$ ($\dagger$).
Then, let 
\[
M=\{x\in\Delta\mid \neg\exists y\in\Delta:~y\prec x\}
~~\mbox{and}~~
B=\{x\in M \mid \forall z\in M :~ |N_x|\leq |N_z|\}\;,
\]
 where $N_z=\{w\in\Delta \mid z\prec w\}$, for every $z\in\Delta$.
 We observe that $M\neq\es$ as $\rs$ is a finite partial
 order (note that \textsc{po{:}2} follows from \textsc{qso{:}1,2}),
 and so also $B\neq\es$.
 Moreover, for every $z\in M$, we have $N_z\neq \es$, by 
 $z\in M$ and ($\dagger$). 
 (Intuitively, from the set $M$ of topologically minimal elements 
 of $\rs$ we distinguished the minimal ones w.r.t.\ the out-degree; the set of such 
 elements is the root of the \textsc{qs}-sequence corresponding to $\rs'$ defined below.) 

 We first observe that $N_x\subseteq N_w$, for all $w\in M$ and $x\in B$ ($\ddagger$). 
 Indeed, if this is not true, then $N_x\setminus N_w\neq \es$. Moreover, by the definition of $B$,
 $|N_x|\leq |N_w|$. Hence, $N_w\setminus N_x\neq \es$. 
 Therefore, there are $y,t\in\Delta$ such that
 $x\prec y$, $w\prec t$, $x\not\prec t$, and $w\not\prec y$,
 yielding to a contradiction with \textsc{qso{:}2}.

 Let $x$ be any element of $B\neq\es$, $N=N_x$, and $P = \Delta\setminus N_x$. 

 Moreover, $(N\times P)\cap\prec=\es$ ($\dagger\dagger$). Indeed, if $a\in N$ and $b\in P$
 such that $a\prec b$, then $x\prec a\prec b$, and so $x\prec b$.
 This, however, means that $b\in N$, a contradiction.

 Both $P$ and $N$ are nonempty disjoint sets ($P\neq\es$ as $x\notin N_x$).
 Hence, 
 by the induction hypothesis, 
 $\rs'=\rs|_P$ and $\rs''=\rs|_N$ are \textsc{qs}-orders.
 
 Suppose that $\rs \neq \rs' \circ \rs''$.
 Then, by Definition~\ref{def-qso} and ($\dagger\dagger$), there are $a\in P$ and $b\in N$ such that $a\not\prec b$.
 By ($\ddagger$), $a\in P\setminus M$ and, since $\prec$ is transitive, there is $c\in M$ such that $c\prec a$. 
 Also, $x\prec b$. Therefore, by \textsc{qso{:}2}, one of the following five cases must hold.
\begin{itemize}
\item $x\prec a$ and $c\prec b$. Then $a\in N$, a contradiction.
\item $x\prec c$ and $x\prec a$. Then $c\notin M$, a contradiction.
\item $c\prec x$ and $c\prec b$. Then $x\notin M$, a contradiction.
\item $b\prec a$ and $x\prec a$. Then $a\in N$, a contradiction.
\item $a\prec b$ and $c\prec b$. Then we obtain a contradiction as we assumed $a\not\prec b$.
\end{itemize}
As a result, $\rs = \rs' \circ \rs''$ and so, by Definition~\ref{def-qso}, $\rs$ is a \textsc{qs}-order. 
\end{proof}
 
We can now show that

\begin{proposition}
\label{prop-vrev}
All stratified orders are 
\textsc{qs}-orders, and all \textsc{qs}-orders are interval orders.
\end{proposition}
\begin{proof}
To show the first inclusion,
let $\rs=\FA{\Delta,\prec}\neq \FA{\es,\es}$ be a stratified order and
$\Delta_1,\dots,\Delta_n$ ($n\geq 1$) be a partition of $\Delta$ such that 
$\prec$ is equal to $\bigcup_{1\leq i<j\leq n}\Delta_i\times\Delta_j$.
Let $\rs_i=\FA{\Delta_i,\es}$, for every $1\leq i \leq n$.
By Definition~\ref{def-qso}, $\rs_i\in \QSO$, for every $1\leq i \leq n$.
Hence, if $n=1$, we obtain $\rs=\rs_1\in\QSO$.
Otherwise, again by Definition~\ref{def-qso},
$\rs=(\dots(\rs_1\circ \rs_2)\circ \dots )\circ\rs_n\in \QSO$.

The second inclusion
follows directly from Theorem~\ref{prop-axioms_qso}($\Rightarrow$) and 
the observation made at the beginning of the proof of
Theorem~\ref{prop-axioms_qso}($\Leftarrow$). 
\end{proof}

\begin{example}
\label{ex-qsorder}
\input{figures/qsorder.tex}
\end{example}

In this section, we introduced quasi-stratified orders which, we believe,
provide a faithful representation of behaviours for a range of concurrent computing 
systems. In fact, we provided two different ways of defining \textsc{qs}-orders, and
the decision which one to use will typically depend on 
the context.

\section{Quasi-stratified sequences}

To explain the rationale behind the adjective in `quasi-stratified order',
we need an auxiliary notion of
a \emph{quasi-stratified sequences} (or \textsc{qs}-sequences).

\begin{definition}[\textsc{qs}-sequence and \textsc{qss}-stratum]
\label{def-qss} 
\emph{Quasi-stratified sequences} (or \textsc{qs}-sequences) 
and \emph{quasi-stratified sequence strata} (or \textsc{qss}-strata)
are the smallest sets, $\QSS$ and $\QSSS$, respectively, 
such that the following hold
(below we also define the domain $\delta_\qss$ of each $\qss\in\QSS$):
\begin{itemize} 

\item
 $\Delta\in\QSS\cap\QSSS$, for every finite nonempty set $\Delta$.
 Moreover, $\delta_\Delta=\Delta$. 

\item
 $\qsss_1\dots\qsss_n\in\QSS$,
 for every $n\geq 2$ and all \textsc{qss}-strata $\qss_1,\dots,\qss_n$    
 with mutually disjoint domains.
 Moreover,
 $\delta_{\qsss_1\dots\qsss_n}=\delta_{\qsss_1}\cup\dots\cup\delta_{\qss_n}$. 
 
 \item 
 $\FA{\Delta,\qsss_1\dots\qsss_n}\in\QSS\cap\QSSS$,
 for every  $n\geq 2$, every finite nonempty set $\Delta$,   
 and all \textsc{qss}-strata $\qss_1,\dots,\qss_n$    
 with mutually disjoint domains satisfying 
 $\Delta\cap\delta_{\qsss_1\dots\qsss_n}=\es$. 
 Moreover,
 $\delta_{\FA{\Delta,\qsss_1\dots\qsss_n}}=
 \Delta\cup\delta_{\qsss_1\dots\qsss_n}$.
\eod 
\end{itemize}   
\end{definition} 
Clearly, \textsc{qss}-strata are also \textsc{qs}-sequences.    
Using a `tree-based methapor', each \textsc{qs}-sequence can be thought of as an ordered forest of trees whose
nodes are disjoint nonempty finite sets; moreover, 
each node either has no direct descendants or the direct descendants form an ordered sequence of 
at least two such trees. 

One can verify that \textsc{qs}-sequences are in a strict correspondence with 
nonempty \textsc{qs}-orders, through the mapping $\mathfrak{G}$ defined thus
(see Definitions~\ref{def-qso} and~\ref{def-qss}): 
\begin{enumerate} 
 \item 
 $\mathfrak{G}(\Delta)=\FA{\Delta,\es}$.

 \item
 $\mathfrak{G}(\qsss_1\dots\qsss_n)=\mathfrak{G}(\qsss_1)\circ \dots\circ\mathfrak{G}(\qsss_n)$.
 
 \item
 $\mathfrak{G}(\FA{\{x_1,\dots,x_k\},\qsss_1\dots\qsss_n})
 =\mathfrak{G}(\qsss_1\dots\qsss_n)[x_1]\dots[x_k]$. 
\end{enumerate} 
Note that in (2) and (3) above the order of application of the operations 
on the rhs does not matter.

We first show the soundness of the mapping $\mathfrak{G}$.

\begin{proposition}
\label{prop-ciue}
$\mathfrak{G}(\qss)\in \QSO$ and $\delta_\qss=\Delta_{\mathfrak{G}(\qss)}$, 
for every \textsc{qs}-sequence $\qss$.   
Moreover,
$\qss$ is a \textsc{qss}-stratum iff
$\mathfrak{G}(\qss)$ is a \textsc{qso}-stratum.
\end{proposition}
\begin{proof}
We proceed by induction on the syntax of \textsc{qs}-sequences.

In the base case, $\qss=\Delta=\{x_1,\dots,x_n\}$ ($n\geq 1$). Then
$\mathfrak{G}(\qss)=(\Delta,\es)=\FA{\es,\es}[x_1]\dots[x_n]\in\QSO$.
Moreover, $\delta_\qss=\Delta=\Delta_{\mathfrak{G}(\qss)}$.

In the inductive case we have the following:

\emph{Case 1:}  $\qss=\qsss_1\dots\qsss_n$ ($n\geq 2$). 
Then, by the induction hypothesis, $\mathfrak{G}(\qsss_i)\in \QSO$
and $\delta_{\qsss_i}=\Delta_{\mathfrak{G}(\qsss_i)}$,
for every $1\leq i\leq n$. Hence, we also have
$\delta_{\mathfrak{G}(\qsss_i)}\cap\delta_{\mathfrak{G}(\qsss_j)}$,
for all $1\leq i< j\leq n$.
Therefore,  by
Definitions~\ref{def-qso} and~\ref{def-qss}, 
$\mathfrak{G}(\qss)=\mathfrak{G}(\qsss_1)\circ\dots\circ\mathfrak{G}(\qsss_n)\in\QSO$
and $\delta_\qss=\Delta_{\mathfrak{G}(\qss)}$.

\emph{Case 2:}  $\qss=\FA{\Delta,\qss'}$
and $\Delta=\{x_1,\dots,x_k\}$, where  $k\geq 1$.
Then, by the induction hypothesis, $\mathfrak{G}(\qss')\in \QSO$
and  $\delta_{\qss'}=\Delta_{\mathfrak{G}(\qss')}$.
Hence, $\Delta\cap\Delta_{\mathfrak{G}(\qss')}=\es$.
Therefore, by Definition~\ref{def-qso},
$\mathfrak{G}(\qss)=\mathfrak{G}(\qss')[x_1]\dots[x_k]\in\QSO$
and $\delta_\qss=\Delta\cup\delta_{\qss'}=
\Delta\cup \Delta_{\mathfrak{G}(\qss')}=\Delta_{\mathfrak{G}(\qss)}$.

To show the second part, suppose that $\qss$ is a \textsc{qss}-stratum.
Then there is $x\in\Delta_{\mathfrak{G}(\qss)}$ such that 
$x\not\prec_{\mathfrak{G}(\qss)}z\not\prec_{\mathfrak{G}(\qss)}x$, for all
$z\in \Delta_{\mathfrak{G}(\qss)}\setminus\{x\}$.
Hence $\mathfrak{G}(\qss)$ is a \textsc{qso}-stratum.
Conversely, if $\mathfrak{G}(\qss)$ is a \textsc{qso}-stratum, then
there is $x\in\Delta_{\mathfrak{G}(\qss)}$ such that 
$x\not\prec_{\mathfrak{G}(\qss)}z\not\prec_{\mathfrak{G}(\qss)}x$, for all
$z\in \Delta_{\mathfrak{G}(\qss)}\setminus\{x\}$.
Hence, $\qss$ has to be of the form $\FA{\Delta,\qss'}$, and so
$\qss$ is a \textsc{qss}-stratum.
\end{proof} 

\begin{proposition}
\label{prop-ciuess}
$\mathfrak{G}:\QSS\to\QSO\setminus\{\FA{\es,\es}\}$
is an injective mapping.   
\end{proposition}
\begin{proof}
By Proposition~\ref{prop-ciue}, $\mathfrak{G}$ is well-defined
(note that $\Delta_{\FA{\es,\es}}=\es$).

Suppose that $\qss,\qss'\in\QSS$ are such that
$\mathfrak{G}(\qss)=\mathfrak{G}(\qss')=\qso=\FA{\Delta,\prec}\in\QSO$. To show
$\qss=\qss'$ we proceed by induction of $|\Delta|$.

In the base case, $|\Delta|=1$, we have $\qss=\qss'=\Delta$.
In the inductive case, suppose that $|\Delta|>1$ and consider two cases.

\emph{Case 1:} There is 
$x\in\Delta$ such that 
$x\not\prec z\not\prec x$, for all
$z\in \Delta \setminus\{x\}$.
Then, $\qss=\FA{\Delta_o,\qss_o}$ and $\qss'=\FA{\Delta'_o,\qss'_o}$.
Moreover, 
$\Delta_o=\{y\in\Delta\mid \forall z\in\Delta\setminus\{y\}: y\not\prec z \not\prec\}
=\Delta'_o$.
We therefore have $\Delta_o =\Delta'_o$ and $\mathfrak{G}(\qss_o)=\mathfrak{G}(\qss'_o)$.
Thus, by the induction hypothesis, $\qss_o=\qss'_o$, and so $\qss=\qss'$. 
\end{proof}
 
To show a one-to-one relationship between nonempty \textsc{qso}-orders 
and \textsc{qs}-sequences, we introduce a mapping 
\[\mathfrak{I}:\QSO\setminus\{\FA{\es,\es}\}\to\QSS\] such that, 
for every $\qso\in\QSO\setminus\{\FA{\es,\es}\}$, $\mathfrak{I}(\qso)$
is defined as follows:
\begin{itemize}
\item 
If $\qso$ is a \textsc{qs}-stratum then 
$\qso=\qso'[x_1]\dots[x_n]$, where $n\geq 1$.
Then $\mathfrak{I}(\qso)=\{x_1,\dots,x_n\}$ if $\qso'=\FA{\es,\es}$,
and
$\mathfrak{I}(\qso)=\FA{\{x_1,\dots,x_n\},\mathfrak{I}(\qso')}$ if $\qso'\neq\FA{\es,\es}$.

\item
If $\qso$ is not a \textsc{qs}-stratum, then there are \textsc{qso}-strata
$\qso_1,\dots,\qso_n$ ($n\geq 2$) such that 
$\qso=\qso_1\circ\dots\circ\qso_n$.
Then 
$\mathfrak{I}(\qso)=\mathfrak{I}(\qso_1)\dots \mathfrak{I}(\qso_n)$. 
\end{itemize} 

\begin{theorem}
\label{th-eic}
$\mathfrak{G}:\QSS\to\QSO\setminus\{\FA{\es,\es}\}$ and 
$\mathfrak{I}:\QSO\setminus\{\FA{\es,\es}\}\to\QSS$
are inverse bijections.
\end{theorem}
\begin{proof}
By Propositions~\ref{prop-ciue} and~\ref{prop-ciuess}, it suffices to show
that $\mathfrak{I}(\mathfrak{G}(\qss))=\qss$, for every $\qss\in\QSS$.
\end{proof}

\begin{remark}[\textsc{qs}-sequences]
\label{rem-3}
Using the notion of \textsc{qs}-sequence, the following is an alternative
definition of a nonempty \textsc{qs}-order: 
\begin{quote}
$\FA{\Delta,\prec}$ is a \textsc{qs}-order if there is a \textsc{qs}-sequence
$\qss=\qsss_1 \dots \qsss_n$ ($n\geq 1$) such that each $\qsss_i$ is a \textsc{qs}-stratum,
$\delta_{\qsss_1},\dots, \delta_{\qsss_n}$
form a partition of $\Delta$, 
and $\prec=\prec_{\mathfrak{G}(\qss)}$.
\end{quote} 
Intuitively, this means that the intervals corresponding to events can be 
partitioned into \textsc{qs}-strata (each being represented by 
a tree whose root is the `base' of a stratum). 
Note that if $\qsss_1=\Delta_1,\dots,\qsss_n=\Delta_n$, then the above   
coincides with
the alternative definition of stratified orders (see Section~\ref{subs-1}).
\end{remark}
 
\section{Quasi-stratified structures}
\label{sect-qscs}

We are now ready to introduce relational structures which will be suitable 
as specifications of concurrent behaviours modelled by \textsc{qs}-orders.
Before that we introduce auxiliary notions.

\begin{definition}[\textsc{csc}-subset]
\label{def-oernier}
Let $\rs=\FA{\Delta,\prec,\sq}$ be a structure and $\Phi\subseteq\Delta$ be a nonempty set.
\begin{itemize}
\item 
$\Phi$ is a \emph{combined strongly connected subset} (or \textsc{csc}-subset)
if 
$\FA{\Phi,(\sq{\cup}\prec)|_{\Phi\times \Phi}}$ 
is a strongly connected directed graph. 
\item 
A \emph{pre-dominant} of $\Phi$ 
is any $x\in\Phi$ such that 
there is no $y\in\Phi$ satisfying $x\prec y$ or $y\prec x$.
We denote this by $x\in \predominants_\rs(\Phi)$.
\end{itemize} 
All \textsc{csc}-subsets of $\rs$ are denoted by $\CSCS(\rs)$, and
all \textsc{csc}-subsets of $\rs$ without
any pre-dominants by $\BCSCS(\rs)$. \EOD
\end{definition}
Intuitively, a \textsc{csc}-subset with pre-dominant $x$ is a potential
candidate for a nested stratum where $x$ belongs to the base (cf.\ Remark~\ref{rem-3}). 
 
\begin{definition}[\textsc{qsa}-structure]
\label{def-qsas} 
A structure 
is \emph{quasi-stratified acyclic} (or \textsc{qs}-acyclic) if it has no 
\textsc{csc}-subset without pre-dominant.
Then, a \emph{quasi-stratified acyclic structure} (or \textsc{qsa}-structure)
is a relational structure which is \textsc{qs}-acyclic.
All \textsc{qsa}-structures are denoted by $\QSAS$.
\EOD
\end{definition}

In the context of the discussion on `acyclicity' in Remark~\ref{rem-2},
in the case of \textsc{qsa}-structures, the role of `forbidden' 
cyclic substructures is played by 
\textsc{csc}-subsets without pre-dominants. The reason why such substructures are 
inadmissible can be explained in the following way. 
Suppose that $\Phi\in\BCSCS(\rs)$ and $\qso$ is a \textsc{qs}-order execution 
in which all events $\Phi$ are executed.
By definition, $\qso$ is either of the form $\qso'[x]$ or of the form 
$\qso'\circ\qso''$. The first case is not possible as then $x$ would be
pre-dominant of $\Phi$, contradicting $\Phi\in\BCSCS(\rs)$.
In the second case, 
all the events in $\Delta_{\qso'}$ must precede all the events in $\Delta_{\qso''}$,
but then, by $\Phi\in\BCSCS(\rs)$ and the strong connectivity property, 
there is a pair of events, $x\in\Delta_{\qso'}$ and $y\in\Delta_{\qso''}$,
such that $y\sq_\rs x$ or $y\prec_\rs x$, yielding to a contradiction with $x\prec_\qso y$. 

We end this section providing some basic properties of relational structures related with quasi-stratified acyclicity.

\begin{proposition}
\label{prop-jngserub}
Let $\rs$ and $\rs'$ be structures such that $\rs\pref\rs'$, and $\Phi\subseteq\Delta_\rs$.
Then:
\begin{enumerate}
 \item 
 $\CSCS(\rs)\subseteq\CSCS(\rs')$ and
 $\BCSCS(\rs)\subseteq\BCSCS(\rs')$.
 \item 
 $\CSCS(\rs|_\Phi)\subseteq\CSCS(\rs)$ and
 $\BCSCS(\rs|_\Phi)\subseteq\BCSCS(\rs)$.
\end{enumerate}
\end{proposition}
\begin{proof}
(1) follows directly from the definitions and the fact that 
adding additional arcs maintains the property of being a strongly connected subgraph,
and the fact that $\prec_\rs\subseteq\prec_{\rs'}$.

(2) follows from from the fact the all relationships
between the elements of $\Phi$
in $\rs|_\Phi$ are the same as in $\rs$.
\end{proof}

\begin{proposition}
\label{prop-radgih}
\textsc{qsa}-structures are projection-closed and intersection-closed.
\end{proposition}
\begin{proof}
It follows directly from Proposition~\ref{prop-jngserub}.
\end{proof}

\section{Maximal QS-structures}
\label{sec-max}

In this section we provide alternative definitions of maximal \textsc{qs}-structures together with some useful properties.

\begin{definition}[\textsc{qsm}-structure]
\label{def-qsms} 
A \emph{quasi-stratified maximal structure} (or \textsc{qsm}-structure) 
is a structure $\qsms = \FA{\Delta,\prec,\sq}$ such that the following hold, for all $x,y,z,t\in \Delta$:
\begin{equation*}
\label{axioms-maxqso}
\begin{array}{llr@{~~~}c@{~~~}l@{~~~~~}l@{~~~}lr@{~~~}c@{~~~}l@{~~~~~}l@{~~~}l}
 \textsc{qsms{:}1}
 &: 
 & x\not\sq x
 \\ 
 \textsc{qsms{:}2}
 & :
 &
 x\prec y
 & \Longleftrightarrow
 & x\sq y\not\sq x 
 \\ 
 \textsc{qsms{:}3}
 & :
 &
 x\neq y
 & \Longleftrightarrow
 & x\prec y \,~\lor~\,y\prec x \,~\lor~\,x\sq y \sq x 
 \\ 
 \textsc{qsms{:}4}
 &:
 &
 x\prec y \land z \prec t
 & \Longrightarrow
 & x\prec t \land z\prec y \,~\lor~\, x\prec z \land x\prec t ~~\lor 
 ~~z\prec x \land z\prec y  
 \,~\lor~\,t\prec y \land z\prec y \,~\lor~\, y\prec t \land x\prec t 

\end{array}
\end{equation*} 
All \textsc{qsm}-structures are denoted by $\QSMS$. 
\EOD
\end{definition}
Note that \textsc{qsms{:}4} corresponds to \textsc{qso{:}2} which
is a key axiom characterising \textsc{qs}-orders. Moreover,
 \textsc{qsm}-structures are projection-closed. 

\begin{example}
\label{ex-max}
\input{figures/max.tex}
\end{example}

The next result clarifies the relationship between 
\textsc{qs}-orders and \textsc{qsm}-structures.
Basically, we will show that $\rho(\QSO)=\QSMS$ (see Section~\ref{subs-2}).

\begin{proposition}
\label{prop-siveiv}
 A structure $\qsas = \FA{\Delta,\prec,\sq}$ belongs to $\QSMS$ iff 
\begin{subequations}
\begin{align}
&\FA{\Delta,\prec}\in\QSO ~~\textit{and}
\label{eq-efi-a} 
\\
&\sq=\{\FA{x,y}\in\Delta\times \Delta\mid y\not\prec x\neq y\}\;.
\label{eq-efi-b} 
\end{align}
\end{subequations} 
\end{proposition}
\begin{proof}
In the proof, we rely on the characterisation of \textsc{qs}-orders given in 
Theorem~\ref{prop-axioms_qso}.

$(\Longrightarrow)$
We observe that Eq.\eqref{eq-efi-a} holds as \textsc{qso{:}1} follows from
\textsc{qsms{:}1,2}, and \textsc{qso{:}2} is 
\textsc{qsms{:}4}.
To prove Eq.\eqref{eq-efi-b}, suppose that $x\sq y$.
Then, $x\neq y$ by \textsc{qsms{:}1}, and $y\not\prec x$ by \textsc{qsms{:}2}.
To show the opposite inclusion, suppose that $x,y\in\Delta$ are such that
$y\not\prec x\neq y$. Then, by \textsc{qsms{:}2}, we have
$x\prec y$ or $x\sq y\sq x$. Moreover, $x\prec y$ implies $x\sq y$,
by \textsc{qsms{:}2}. 

$(\Longleftarrow)$
We observe that \textsc{qsms{:}1} follows from Eq.\eqref{eq-efi-b}.
To show \textsc{qsms{:}2}, suppose that $x\prec y$. 
Then, by $\FA{\Delta,\prec}\in\QSO$, $x\neq y$ and $y\not\prec x$.
Hence, by Eq.\eqref{eq-efi-b}, $x\sq y\not\sq x$.
To show \textsc{qsms{:}3}, suppose that $x,y\in\Delta$ and $x\neq y$.
If $x\not\prec y\not\prec x$ then, by Eq.\eqref{eq-efi-b}, $x\sq y \sq x$.
Finally, \textsc{qsms{:}4} is 
\textsc{qso{:}2}. 
\end{proof}
 
We are now in a position to show that 
\textsc{qsa}-structures are a suitable specification device 
for behaviours represented by \textsc{qs}-orders.

\begin{theorem}
\label{prop-grtger}
$\QSMS =\QSAS^\sat$.
\end{theorem}
\begin{proof}
$(\subseteq)$
Let $\qsms=\FA{\Delta,\prec,\sq}\in\QSMS$ and $\qso=\FA{\Delta,\prec}$. 
By Proposition~\ref{prop-siveiv}, $\qso\in\QSO$.
We first show that $\qsms\in\QSAS$, proceeding by induction on structure of
$\qso$ (see Definition~\ref{def-qso}). 
If $\qso=\FA{\es,\es}$ then, clearly, $\qsms\in\QSAS$.
In the inductive step, we consider two cases.
 
\emph{Case 1:} 
$\qso=\qso'[x]$, where $\qso'\in\QSO$. Then, by the induction hypothesis, 
$\qsms'=\qsms|_{\Delta\setminus\{x\}}\in\QSAS$.
We then observe that, by Proposition~\ref{prop-siveiv}, 
$\CSCS(\qsms)\setminus\CSCS(\qsms')=
\{\Phi\in\CSCS(\qsms) \mid x\in\Phi \}$.
And so, $x\in\predominants_\qsms(\Phi)$, for every 
$\Phi\in\CSCS(\qsms)\setminus\CSCS(\qsms')$. Hence $\qsms\in\QSAS$

\emph{Case 2:} 
$\qso=\qso'\circ\qso''$, where $\Delta_{\rs'}\neq\es\neq\Delta_{\rs''}$.
Then, by the induction hypothesis, we have $\qsms'=\qsms|_{\Delta_{\rs'}}\in\QSAS$
and $\qsms''=\qsms|_{\Delta_{\rs''}}\in\QSAS$.
Moreover, $(\prec\cup\sq)\cap (\Delta_{\rs''}\times\Delta_{\rs'})=\es$,
which means that 
$\CSCS(\qsms)=\CSCS(\qsms')\cup\CSCS(\qsms'')$.
 
Hence $\qsms\in\QSAS$. Moreover, by Proposition~\ref{prop-siveiv}, for all $x\neq y\in\Delta$,
$x\not\prec y$ implies $\{x,y\}\in \BCSCS(\qsms{\adding{x}{\prec}{y}})$ and 
$x\not\sq y$ implies $\{x,y\}\in \BCSCS(\qsms{\adding{x}{\sq}{y}})$.
Hence $\qsms\in\QSAS^\sat$.

$(\supseteq)$
Since \textsc{qsa}-structures are finite and
we have already proved that $\QSMS\subseteq\QSAS$, it suffices to show that $\Ext{\QSMS}{\qsas}\neq\es$,
for every $\qsas=\FA{\Delta,\prec,\sq}\in\QSAS$.
We proceed by induction on the size of the domain of $\qsas$.
Clearly, $\FA{\es,\es,\es}\in\Ext{\QSMS}{\FA{\es,\es,\es}}$
and $\FA{\{x\},\es,\es}\in\Ext{\QSMS}{\FA{\{x\},\es,\es}}$. In
the inductive step $|\Delta|>1$ and we consider two cases.

\emph{Case 1:} $\Delta\in\CSCS(\qsas)$. Then, by $\qsas\in\QSAS$, there is $x\in\predominants_\qsas(\Delta)$.
Let $\Delta'=\Delta\setminus\{x\}$.
As \textsc{qsa}-structures are projection-closed, by Proposition~\ref{prop-radgih}, $\qsas|_{\Delta'}\in\QSAS$.
Hence, by the induction hypothesis, there is 
$\FA{\Delta',\prec',\sq'}\in\Ext{\QSMS}{\qsas|_{\Delta'}}$.
It is straightforward to check that 
$ 
\FA{\Delta,\prec',\sq'\cup(\Delta'\times\{x\})\cup(\{x\}\times\Delta')}
$
belongs to $\Ext{\QSMS}{\qsas}$.

\emph{Case 2:} $\Delta\notin\CSCS(\qsas)$. Then $\Delta$ can be partitioned
as $\Delta'\uplus \Delta''$ in such a way that 
we have $(\Delta''\times\Delta')\cap(\prec\cup\sq)=\es$.
As \textsc{qsa}-structures are projection-closed, $\qsas|_{\Delta'}\in\QSAS$ and $\qsas|_{\Delta''}\in\QSAS$.
Hence, by the induction hypothesis, there are
$\qsms'\in \Ext{\QSMS}{\qsas|_{\Delta'}}$ and 
$\qsms''\in\Ext{\QSMS}{\qsas|_{\Delta''}}$.
It is straightforward to check that 
$
\FA{\Delta,(\prec_{\qsms'}\cup\prec_{\qsms''})\cup(\Delta'\times\Delta''),
(\sq_{\qsms'}\cup\sq_{\qsms''})\cup(\Delta'\times\Delta'')}
$
belongs to $\Ext{\QSMS}{\qsas}$. 
\end{proof}

Having at least one \textsc{qsm}-structure extension yields an alternative 
characterisation of \textsc{qsa}-structures.

\begin{proposition}
\label{prop-regrehere}
A structure is a \textsc{qsa}-structure iff one of its extensions 
is a \textsc{qsm}-structure.
\end{proposition}
\begin{proof}
Let $\rs=\FA{\Delta,\prec,\sq}$ be a structure. 

($\Longrightarrow$)
If $\rs\in\QSAS$, then $\Ext{\QSMS}{\rs}\neq\es$ 
follows from the finiteness of $\rs$ and Theorem~\ref{prop-grtger}.

($\Longleftarrow$) 
Suppose that $\rs\notin\QSAS$. As $\Ext{\QSMS}{\rs}\neq\es$ and 
\textsc{qsms{:}1,2}, $\rs$ is a relational structure. 
Hence there is $\Phi\in\CSCS(\rs)$
such that $\predominants_\rs(\Phi)=\es$.
Suppose then that $\qsms\in\Ext{\QSMS}{\rs}$. 
Then $\Phi\in\BCSCS(\qsms)$, yielding a contradiction 
with $\QSMS\subseteq\QSAS$ proved in Theorem~\ref{prop-grtger}. 
\end{proof}

The next result captures some cases where extending a \textsc{csa}-structure
may (or may not) lead outside \textsc{csa}-structures. 
This way we provide a recipe to properly saturate non-maximal \textsc{qsa}-structures. 

\begin{proposition}
\label{prop-rgih}
Let $\qsas = \FA{\Delta,\prec,\sq}\in\QSAS$ and
$x\neq y\in\Delta$. Then:

\begin{enumerate}
 
\item
$x\prec y$ implies $\qsas{\adding{y}{\sq}{x}}\notin\QSAS$.

\item
$x\sq y$ implies $\qsas{\adding{y}{\prec}{x}}\notin\QSAS$.

\item 
$\qsas{\adding{x}{\prec}{y}}\in\QSAS$ 
or 
$\qsas{\adding{y}{\sq}{x}}\in\QSAS$.

\item 
$\qsas{\adding{x}{\prec}{y}}\notin\QSAS$ 
implies 
$\qsas{\adding{y}{\sq}{x}}\in\QSAS$
and $\satmap_\QSAS(\qsas{\adding{y}{\sq}{x}})=\satmap_\QSAS(\qsas)$.

\item 
$\qsas{\adding{x}{\sq}{y}}\notin\QSAS$ 
implies 
$\qsas{\adding{y}{\prec}{x}}\in\QSAS$
and $\satmap_\QSAS(\qsas{\adding{y}{\prec}{x}})=\satmap_\QSAS(\qsas)$.
 
\end{enumerate}
\end{proposition}
\begin{proof} 
(1) follows from $\{x,y\}\in\BCSCS(\qsas{\adding{y}{\sq}{x}})$.

(2) follows from $\{x,y\}\in\BCSCS(\qsas{\adding{y}{\prec}{x}})$.
 
(3) 
 We will proceed by contradiction. Suppose that 
 $\rs=\qsas{\adding{x}{\prec}{y}}\notin\QSAS$ and
 $\rs'=\qsas{\adding{y}{\sq}{x}}\notin\QSAS$.
 As both $\rs$ and $\rs'$ are relational structures,
 there are $\Phi\in\BCSCS(\rs)$ and $\Psi\in\BCSCS(\rs')$.
 From $\qsas\in\QSAS$ it follows that $x,y\in \Phi\cap\Psi$. Moreover, 
 $\Phi\in\BCSCS(\rs)$ implies
 that there is a path, using $\prec$ and $\sq$, from $y$ to $x$ in $\qsas$ going through the elements of $\Phi$,
 and
 $\Psi\in\BCSCS(\rs')$ implies
 that there is a path, using $\prec$ and $\sq$, from $y$ to $x$ in $\qsas$ going through the elements of $\Psi$.
 Hence $\Phi\cup\Psi\in\CSCS(\qsas)$, and so from $\qsas\in\QSAS$ it follows 
 that there is $z\in\predominants_\qsas(\Phi\cup\Psi)$.
 It then follows that $z\neq x$, as otherwise we would have 
 $x\in\predominants_{\rs'}(\Psi)$, contradicting $\Psi\in\BCSCS(\rs')$.
 Similarly, $z\neq y$.
 Then, if $z\in\Phi$,
 we obtain a contradiction with $\Phi\in\BCSCS(\rs)$, and 
 if $z\in\Psi$,
 we obtain a contradiction with $\Psi\in\BCSCS(\rs')$. 
 Hence $\rs\in\QSAS$ or $\rs'\in\QSAS$.

 (4)
 By part (3), $\qsas{\adding{y}{\sq}{x}}\in\QSAS$. 
 Clearly, $\satmap_\QSAS(\qsas{\adding{y}{\sq}{x}})\subseteq\satmap_\QSAS(\qsas)$.
 Suppose that $\qsms\in \satmap_\QSAS(\qsas)\setminus \satmap_\QSAS(\qsas{\adding{y}{\sq}{x}})$.
 Then we must have $y\not\sq_\qsms x$, and so $x \prec_\qsms y$.
 As a result, $\qsms\in \satmap_\QSAS(\qsas{\adding{x}{\prec}{y}})$.
 Hence, by Proposition~\ref{prop-regrehere}, $\qsas{\adding{x}{\prec}{y}}\in\QSAS$, a contradiction.

 (5) Similarly as for part (4).
\end{proof}

Note that Proposition~\ref{prop-rgih}(4) or Proposition~\ref{prop-rgih}(5) 
implies  Proposition~\ref{prop-rgih}(3), but we preferred to carry out the proof
in smaller steps.  

\section{Closed QS-structures}
\label{sec-inv}
 
\textsc{qsa}-structures can be regarded as 
suitable specifications of \textsc{qs}-order behaviours.
In this section, we will look closer at \textsc{qsa}-structures
which are particularly suitable for, \eg property verification or 
causality analysis in concurrent behaviours.
In the general theory expounded in~\cite{DBLP:series/sci/2022-1020}, a special place is occupied 
by the following structures which enjoy 
properties of this kind.

\begin{definition}[closed \textsc{qsa}-structure]
\label{def-closed-qsas} 
A \textsc{qsa}-structure $\qsas\in\QSAS$ is \emph{closed} if     
\[\qsas=\FA{\Delta_\qsas,\prec_{\satmap_\QSAS(\qsas)},\sq_{\satmap_\QSAS(\qsas)}}\}\;\]
where
\[
\begin{array}{rcl}
\makebox[3cm][r]{\mbox{$\prec_{\satmap_\QSAS(\qsas)}$}} 
&=&\bigcap\{\prec_\qsms\mid\qsms\in\satmap_\QSAS(\qsas)\}
\\
\sq_{\satmap_\QSAS(\qsas)}&=&\bigcap\{\sq_\qsms\mid\qsms\in\satmap_\QSAS(\qsas)\}\;. 
\end{array}
\]
All closed \textsc{qsa}-structures are denoted by $\QSAS^\closed$. 
\EOD
\end{definition}
 
From the general results presented in~\cite{DBLP:series/sci/2022-1020} (Prop.7.5), 
and the fact that $\QSAS$ is a 
set of intersection closed relational structures, it follows that there is a \emph{structure closure}
mapping $\cls_\QSAS:\QSAS\to\QSAS^\closed$
such that the following hold.
\begin{itemize}
\item 
 A closed \textsc{qsa}-structure is the largest \textsc{qsa}-structure 
 which generates a given set of \textsc{qsm}-structures, \ie
 for all $\rs\in\QSAS^\closed$ and $\qsas\in\QSAS$:
\begin{equation}
 \rs\pref \qsas~\wedge~ \rs \neq\qsas
 \implies
 \satmap_\QSAS(\qsas)\neq\satmap_\QSAS(\rs)\;.
\end{equation}

\item
 $\cls_\QSAS$ is a closure mapping in the usual sense, \ie
 for every $\qsas\in\QSAS$:
\begin{equation} 
 \qsas\pref\cls_\QSAS(\qsas) =\cls_\QSAS(\cls_\QSAS(\qsas))\;.
\end{equation}

\item
 Structure closure preserves the
 generated \textsc{qsm}-structures, \ie
 for every $\qsas\in\QSAS$:
\begin{equation}  
 \satmap_\QSAS(\qsas)=\satmap_\QSAS(\cls_\QSAS(\qsas)) \;.
\end{equation}

\item
 Generating the same \textsc{qsm}-structures is equivalent to having the same closure, \ie
 for all $\qsas,\qsas'\in\QSAS$:  
\begin{equation}
\label{hhddh} 
\satmap_\QSAS(\qsas)=\satmap_\QSAS(\qsas') \iff \cls_\QSAS(\qsas)=\cls_\QSAS(\qsas')\;.
\end{equation}

\item
 Structure closure is monotonic, \ie for all $\qsas,\qsas'\in\QSAS$: 
\begin{equation}
\qsas\pref\qsas'\implies \cls_\QSAS(\qsas)\pref\cls_\QSAS(\qsas') \;.
\end{equation}

\item
The structure closure of a given $\qsas\in\QSAS$ is given by:
\begin{equation}
    \cls_\QSAS(\qsas) = \FA{\Delta_\qsas,\prec_{\satmap_\QSAS(\qsas)},\sq_{\satmap_\QSAS(\qsas)}}\;.
\end{equation}

\end{itemize}

The above, generic, characterisation of closed \textsc{qsa}-structures 
is not a practical one because it directly refers
to potentially (very) large set of maximal extensions of \textsc{qsa}-structures. 
To address this, we will now present an alternative algebraic characterisation 
of $\QSAS^\closed$ and $\cls_\QSAS$. 
We start by introducing an alternative definition of $\QSAS^\closed$
which will then be shown to coincide with $\QSAS^\closed$. 
 
\begin{definition}[\textsc{qsc}-structure]
\label{def-rrrrrs5}
A \emph{quasi-stratified closed structure} (or \textsc{qsc}-structure) 
is a structure 
$\qscs = \FA{\Delta,\prec,\sq}$ such that the following hold, for all $x\neq y\in \Delta$:
\begin{equation*}
\label{axioms-nevnr}
\begin{array}{llr@{~~~}c@{~~~}r@{~~~~~}l@{~~~}lr@{~~~}c@{~~~}l@{~~~~~}l@{~~~}l}
 \textsc{qscs{:}1}
 &:
 & 
 x\not\sq x \not\prec x 
 \\ 
 \textsc{qscs{:}2}
 &:
 &
 x\prec y
 & \Longrightarrow
 & y\not\sq x 
 \\ 
 \textsc{qscs{:}3}
 &:
 &
 \BCSCS(\qscs{\adding{x}{\prec}{y}}) \neq\es
 & \Longrightarrow
 & y\sq x 
 \\ 
 \textsc{qscs{:}4}
 &:
 &\BCSCS(\qscs{\adding{x}{\sq}{y}}) \neq\es
 & \Longrightarrow
 & y\prec x 
\end{array}
\end{equation*} 
All \textsc{qsc}-structures are denoted by $\QSCS$. \EOD
\end{definition}
In the rest of this section, we show that $\QSCS=\QSAS^\closed$.

\begin{example}
\label{ex-closed}
\input{figures/closed.tex}
\end{example}

By Proposition~\ref{prop-jngserub}(1), \textsc{qsc}-structures are projection-closed
and, as shown below, are positioned in-between $\QSMS$ and $\QSAS$.

\begin{proposition} 
\label{prop-inclusions}
$\QSMS\subseteq\QSCS\subseteq\QSAS$.
\end{proposition}

\begin{proof}
($\QSMS\subseteq\QSCS$)
 Let $\qsms = \FA{\Delta,\prec,\sq}\in\QSMS$.
 Then \textsc{qscs{:}1} and \textsc{qscs{:}2} hold for $\qsms$, by \textsc{qsms{:}1}
 and \textsc{qsms{:}2}, respectively.
 To show \textsc{qscs{:}3} for $\qsms$, suppose that $x\neq y\in\Delta$ 
 and
 $\BCSCS(\qsms{\adding{x}{\prec}{y}})$.
 Then, by Theorem~\ref{prop-grtger}, $x\not\prec y$. Hence, by \textsc{qsms{:}2,3}, 
 we obtain $y\sq x$.
 Showing \textsc{qscs{:}4} for $\qsms$ is similar. 
 As a result, $\QSMS\subseteq\QSCS$. 
 
($\QSCS\subseteq\QSAS$)
By \textsc{qscs{:}1}, it suffices to show that if 
$\qscs= \FA{\Delta,\prec,\sq}\in\QSCS$ then $\BCSCS(\qscs)=\es$. 
Indeed, suppose that $\Phi\in\BCSCS(\qscs)$.
 As $|\Phi|=1$ implies $\predominants_\qscs(\Phi)=\Phi\neq \es$,
 we have $|\Phi|\geq 2$. Hence, since $\Phi\in\BCSCS(\qscs)$,
 there are $x\neq y\in \Delta$ such that $x\prec y$.
 Hence, by \textsc{qscs{:}3}, we obtain $y\sq x$, a contradiction with \textsc{qscs{:}2}. 
\end{proof}

The next result identifies two ways in which 
\textsc{qscs}-structures can be extended to \textsc{qsas}-structures.

\begin{proposition}
\label{prop-rju}
Let $\qscs=\FA{\Delta,\prec,\sq}\in\QSCS$ and $x\neq y\in\Delta$. Then: 
\begin{enumerate}

\item $x\not\prec y$ implies $\qscs{\adding{y}{\sq}{x}}\in\QSAS$.

\item $x\not\sq y$ implies $\qscs{\adding{y}{\prec}{x}}\in\QSAS$.
\end{enumerate}
\end{proposition}

\begin{proof} 
(1) 
If $\qscs{\adding{y}{\sq}{x}}\notin\QSAS$, then 
$\BCSCS(\qscs{\adding{y}{\sq}{x}})\neq\es$. 
Hence, by \textsc{qscs{:}4}, we obtain $x\prec y$, yielding a contradiction.

(2)
If $\qscs{\adding{y}{\prec}{x}}\notin\QSAS$, then 
$\BCSCS(\qscs{\adding{y}{\prec}{x}})\neq\es$. 
Hence, by \textsc{qscs{:}3}, we obtain $x\sq y$, yielding a contradiction. 
\end{proof}

We will now define a mapping which is intended to provided an alternative
definition of the structure closure mapping for \textsc{qsa}-structures.
To start with, let 
$\mathfrak{F}:\QSAS\to\QSAS$ be a mapping such that, for every $\qsas=\FA{\Delta,\prec,\sq}\in\QSAS$: 
\[
\begin{array}{rcl}
\mathfrak{F}(\qsas)&=&\FA
{\Delta,
\prec\cup\prec',
\sq\cup \sq'
} ~~~~~~~ 
\mbox{where:}
\\[2mm]
\prec'&=&
\{\FA{y,x}\in\Delta\times\Delta\mid x\neq y
~\wedge~ y\not\prec x~\wedge ~
\BCSCS(\qscs{\adding{x}{\sq}{y}})\neq \es\}
\\
\sq'&=&
\{\FA{y,x}\in\Delta\times\Delta\mid x\neq y~\wedge~ 
y\not\sq x~\wedge~\BCSCS(\qscs{\adding{x}{\prec}{y}})\neq \es\}\;.
\end{array}
\]

\begin{proposition}
\label{prop-fheif}
 $\mathfrak{F}$ is a well-defined function such that, for every $\qsas\in\QSAS$:
 \begin{enumerate}
 \item 
$\qsas\pref\mathfrak{F}(\qsas)\pref \cls_\QSAS(\qsas)$.
 \item
 $\mathfrak{F}(\qsas)=\qsas$ iff $\qsas\in\QSCS$.
 \end{enumerate}
\end{proposition}

\begin{proof}
(1)
Suppose that $x\prec'y$. Then $x\neq y$ and 
$\BCSCS(\qsas{\adding{y}{\sq}{x}})\neq\es$.
Hence, by 
Proposition~\ref{prop-rgih}(5), 
\[
\qsas{\adding{x}{\prec}{y}}\in\QSAS 
~~\mbox{and}~~  \satmap_\QSAS(\qsas{\adding{x}{\prec}{y}})=\satmap_\QSAS(\qsas)\;.
\]
Thus, by Eq.\eqref{hhddh}, 
we have
\[
\qsas{\adding{x}{\prec}{y}}\pref \cls_\QSAS(\qsas{\adding{x}{\prec}{y}})=\cls_\QSAS(\qsas)\;. 
\]
Similarly, we may show that $x\sq'y$ implies 
$\qsas{\adding{x}{\sq}{y}}\pref \cls_\QSAS(\qsas)$.
As result, we have $\qsas\pref\mathfrak{F}(\qsas)\pref \cls_\QSAS(\qsas)$.
Thus $\mathfrak{F}(\qsas)\in\QSAS$, and so $\mathfrak{F}$ is well-defined.

(2) Follows directly from Definition~\ref{def-rrrrrs5} and the definition of $\mathfrak{F}$.
\end{proof}

\begin{theorem} 
\label{th-trwrt} 
 The mapping $\qsasTOqscs:\QSAS\to\QSCS$, 
 for every $\qsas\in\QSAS$ given by: 
 \[
 \qsasTOqscs(\qsas)
 =\mathfrak{F}^{2\cdot|\Delta_\qsas|^2}(\qsas)
 \]
is the structure closure of $\QSAS$, \ie 
$\QSCS=\QSAS^\closed$ and 
$\qsasTOqscs=\cls_\QSAS$. 
\end{theorem}
 
\begin{proof} 
The number of relationships which the successive applications 
of $\mathfrak{F}$ can add is equal to $2\cdot|\Delta_\qsas|^2$.
Hence, $\mathfrak{F}(\qsasTOqscs(\qsas))=\qsas$,
Hence, by Proposition~\ref{prop-fheif}(2), $\qsasTOqscs(\qsas)\in\QSCS$
and is $\qsasTOqscs$ is well-defined.

The theorem is proven by applying a general result 
--- allowing one to deal with the closure operation and closed relational structures 
without referring to intersections of the
maximal extensions of structures --- 
formulated in~\cite{DBLP:series/sci/2022-1020} as Prop.7.8 
and reproduced below.

\begin{quote}  
\emph{
Let $\RR$ be an intersection-closed set of relational structures,
$\SSS\subseteq \RR$, and $f:\RR\to\SSS$
be a monotonic (\ie $\rs\pref\rs'\implies f(\rs)\pref f(\rs')$)
and 
non-decreasing (\ie $\rs\pref f(\rs)$) function.
Moreover, for all $\rs=\FA{\Delta,\prec,\sq}\in\SSS$ and $x\neq y\in\Delta$,
we have:
\begin{itemize}
\item 
$f(\rs)\pref \rs$.
\item 
If $x\not \prec y$ (or $x\not \sq y$), then there is 
$\overline{\rs}\in \satmap_\RR(\rs)$
such that $x\not \prec_{\overline{\rs}}y$ (resp.\ $x\not \sq_{\overline{\rs}} y$).
\end{itemize}
Then $f$ is the structure closure of $\RR$,
\ie $\SSS=\RR^\closed$ and $f(\rs)=\cls_\RR(\rs)$, for every $\rs\in\RR$. 
}
\end{quote}

We can apply the above result after observing the following: 
\begin{itemize} 
 \item 
 $\QSAS$ are intersection-closed by Proposition~\ref{prop-radgih}.
 
\item 
 $\QSCS\subseteq\QSAS$ by Proposition~\ref{prop-inclusions}.

\item
 $\qsasTOqscs$ is monotonic 
 (\ie $\qsas\pref\qsas'\implies \qsasTOqscs(\qsas)\pref \qsasTOqscs(\qsas')$)
 by Proposition~\ref{prop-jngserub}.

\item 
 $\qsasTOqscs$ is non-decreasing (\ie $\qsas\pref \qsasTOqscs(\qsas)$)
 by Proposition~\ref{prop-fheif}(1).

\item 
 For every $\qscs\in\QSCS$,
we have 
$\qsasTOqscs(\qscs)\pref\qscs$, as $\qsasTOqscs(\qscs)=\qscs$ because in such a case $\prec'=\sq'=\es$
in the definition of $\mathfrak{F}$.
 
\item 
 For all $\qscs=\FA{\Delta,\prec,\sq}\in\QSCS$ and $x\neq y\in\Delta$,
if $x\not \prec y$ (or $x\not \sq y$), then there is 
$\qsms\in \satmap_\QSAS(\qscs)$
such that $x\not \prec_\qsms y$ (resp.\ $x\not \sq_\qsms y$).
This follows from Theorem~\ref{prop-grtger} and Proposition~\ref{prop-rju}.
\end{itemize}
Hence, $\QSCS=\QSAS^\closed$ and 
$\qsasTOqscs=\cls_\QSAS$.
\end{proof} 

Theorem~\ref{th-trwrt} is a generalisations of Szpilrajn's Theorem~\cite{Szp30} 
stating that a partial order is the intersection of all
its total order extensions.
Note the complete algebraic characterisation of $\QSAS^\closed$ given in
Definition~\ref{def-rrrrrs5} admits redundant axioms. 
Although it is an open problem what might constitute a minimal 
and complete axiomatisation, the next result can be seen as a precursor for the development of
a minimal complete axiomatisation of closed \textsc{qsa}-structures.

\begin{proposition}
\label{prop-nevnr}
Let $\qscs = \FA{\Delta,\prec,\sq}\in\QSCS$ 
and $x,y,z,t\in \Delta$.
Then:
\begin{subequations}
\begin{align} 
\makebox[5cm][r]{\mbox{$x\prec y$}} 
~\Longrightarrow~ 
x\sq y 
\label{props-nevnr-b} 
\\
\makebox[5cm][r]{\mbox{$x\prec y \sq z \prec t$}} 
~\Longrightarrow~ 
x\prec t 
\label{props-nevnr-c}
\\
\makebox[5cm][r]{\mbox{$x\sq y\prec z ~\vee~ x\prec y\sq z$}} 
~\Longrightarrow~ 
x\sq z 
\label{props-nevnr-d}
\\
\makebox[5cm][r]{\mbox{$x\sq y\prec z\sq t\neq x$}} 
~\Longrightarrow~ x\sq t 
\label{props-nevnr-e}
\\
\makebox[5cm][r]{\mbox{$x\sq z\prec y\sq x $}} 
~\Longrightarrow~ z\sq x\sq y 
\label{props-nevnr-f}
\\
\makebox[5cm][r]{\mbox{$t\prec x\sq z\prec y\sq x$}} 
~\Longrightarrow~ 
t\prec z~\wedge~t\prec y 
\label{props-nevnr-g}
\\
\makebox[5cm][r]{\mbox{$ x\sq z\prec y\sq x\prec t$}} 
~\Longrightarrow~ 
z\prec t ~\wedge~ y\prec t 
\label{props-nevnr-h} 
\\
\makebox[5cm][r]{\mbox{$x\neq t\sq y \sq x \sq z\prec y$}} 
~\Longrightarrow~ 
t\sq x 
\label{props-nevnr-i} 
\\
\makebox[5cm][r]{\mbox{$z\prec y\sq x\sq z\sq t\neq x$}}
~\Longrightarrow~ 
x\sq t 
\label{props-nevnr-j} 
\\
\makebox[5cm][r]{\mbox{$x\prec z\sq y ~\wedge~ x\sq t\prec y$}} 
~\Longrightarrow~ 
x\prec y 
\label{props-nevnr-l} 
\end{align}
\end{subequations} 
\end{proposition}

\begin{proof} 
The proofs of (\ref{props-nevnr-h}) and (\ref{props-nevnr-j}) are
similar to the proofs of (\ref{props-nevnr-g}) and (\ref{props-nevnr-i}), respectively.
\begin{itemize}[leftmargin=+.4in]
\item[(\ref{props-nevnr-b})] 
 By \textsc{qscs{:}1}, $x\neq y$. 
 Then 
 $\{x,y\}\in\BCSCS(\qscs{\adding{y}{\prec}{x}})$. 
 Hence, by \textsc{qscs{:}3}, $x\sq y$.
 
\item[(\ref{props-nevnr-c})] 
 By Proposition~\ref{prop-inclusions}, $x\neq t$. Then 
 $\{x,y,z,t\} \in\BCSCS(\qscs{\adding{t}{\sq}{x}})$. 
 Hence, by \textsc{qscs{:}4}, $x\prec t$.
 
\item[(\ref{props-nevnr-d})] 
 By \textsc{qscs{:}2}, $x\neq z$. 
 Then 
 $\{x,y,z\} \in\BCSCS(\qscs{\adding{z}{\prec}{x}})$. 
 Hence, by \textsc{qscs{:}3}, $x\sq z$.
 
\item[(\ref{props-nevnr-e})] 
 Then
 $\{x,y,z,t\} \in\BCSCS(\qscs{\adding{t}{\prec}{x}})$. 
 Hence, by \textsc{qscs{:}3}, $x\sq t$.
 
\item[(\ref{props-nevnr-f})] 
 By \textsc{qscs{:}1}, $x\neq y$.
 Then 
 $\{x,z,y\}\in \BCSCS(\qscs{\adding{y}{\prec}{x}})$. 
 Hence, by \textsc{qscs{:}3}, $x\sq y$. \\
 Similarly, we can show that $z\sq x$.
 
\item[(\ref{props-nevnr-g})] 
 By \textsc{qscs{:}2}, $t\neq z$. 
 Then
 $\{x,y,z,t\}\in\BCSCS(\qscs{\adding{z}{\sq}{t}})$. 
 Hence, by \textsc{qscs{:}4}, $t\prec z$.\\
 Similarly, we can show that $t\prec y$.
 
\item[(\ref{props-nevnr-i})] 
 Since $t\neq x$,
 $\{x,y,z,t\}\in\BCSCS(\qscs{\adding{x}{\prec}{t}})$. 
 Hence, by \textsc{qscs{:}3}, $t\sq x$.
 
\item[(\ref{props-nevnr-l})] 
 By \textsc{qscs{:}2}, $x\neq y$.
 Then
 $\{x,y,z,t\} \in\BCSCS(\qscs{\adding{y}{\sq}{x}})$. 
 Hence, by \textsc{qscs{:}3}, $x\prec y$. 
\end{itemize}
\end{proof}

For example, Eqs.\eqref{props-nevnr-b}--\eqref{props-nevnr-e}
are part of the axiomatisation of complete relational structures in~\cite{JKKM-PN2024}.
We hypothesise that \textsc{qscs{:}1,2} together with
Eqs.\eqref{props-nevnr-b}--\eqref{props-nevnr-l} constitute 
a minimal and complete set of axioms which use up to four events.

\begin{remark}
Consider again the \textsc{qsa}-structure from Example~\ref{ex-max}, and analyse all its \textsc{qsm}-structure extensions (saturations) together with their representations
both as set of intervals and as \textsc{qs}-sequences (see Figure~\ref{fig-allmax}).
It is easy to check that in all the saturations we have an additional relationship $a\prec d$, which is consistent with the definition of the closure and Example~\ref{ex-closed}.
Moreover, relationship between $c$ and $d$ varies 
depending on the saturation, and takes all possible values. 
The set of \textsc{qsm}-structure extensions additionally verifies the generalisation of Szpilrajn's Theorem, demonstrating that the closure of a \textsc{qsa}-structure is the intersection of all its saturations.
\input{figures/allmax.tex}
\eod
\end{remark}

We end providing additional properties of \textsc{qsc}-structures.

\begin{proposition}
\label{prop-regergeee}
Let $\qscs= \FA{\Delta,\prec,\sq}\in\QSCS$, $\Phi\in\CSCS(\qscs)$ and $x,y,z\in\Delta$.
Then:
\begin{enumerate}
\item 
$x\sq y\prec z\sq x$ 
 implies
 $\{x,y,z\}\in\CSCS(\qscs)$ and $\{x\}=\predominants_\qscs(\{x,y,z\})$.
 \item 
 $\predominants_\qscs(\Phi) = \{x,y\} \land x\neq y$
 implies 
 $x\sq y \sq x$.
\end{enumerate} 
\end{proposition}
\begin{proof} 
(1)
Clearly, $\{x,y,z\}\in\CSCS(\qscs)$, and 
 $y,z\notin\predominants_\qscs(\{x,y,z\})$ since $y\prec z$. By 
 Eq.\eqref{props-nevnr-d}, $y\sq x$ and $z\sq x$.
 Hence, by \textsc{qscs{:}2}, $x\in \predominants_\qscs(\{x,y,z\})$, 
 and so $\{x\}=\predominants_\qscs(\{x,y,z\})$.

 (2) 
 We have 
 $\Phi \in\BCSCS(\qscs{\adding{x}{\prec}{y}}) $. 
 Hence, by \textsc{qscs{:}3}, we obtain $y\sq x$.
 Similarly, we can show that $x\sq y$.
\end{proof}

\begin{proposition}
\label{prop-rvev}
Let $\qscs=\FA{\Delta,\prec,\sq}\in\QSCS$ and $x\neq y\in\Delta$
be such that $x\not\prec y\not\sq x$. 
Then both
$\qscs{\adding{y}{\sq}{x}}$ and
$\qscs{\adding{x}{\prec}{y}}$ belong to $\QSAS$, and the following hold: 
\begin{equation}
\label{eq-jddj}
\begin{array}{lcl}
 \satmap_\QSAS(\qscs)\setminus\satmap_\QSAS(\qscs{\adding{y}{\sq}{x}})&\neq&\es\\
 \satmap_\QSAS(\qscs)\setminus\satmap_\QSAS(\qscs{\adding{x}{\prec}{y}})&\neq&\es\;.\\
\end{array}
\end{equation} 
\end{proposition}

\begin{proof} 
By Proposition~\ref{prop-rju}, $\rs=\qscs{\adding{y}{\sq}{x}}\in\QSAS$
and $\rs'=\qscs{\adding{x}{\prec}{y}}\in\QSAS$. 
Then, by \textsc{qsms{:}2},
$\satmap_\QSAS(\rs)\cap\satmap_\QSAS(\rs')=\es$.
Hence, 
\[
\begin{array}{lllll}
\es&\neq&\satmap_\QSAS(\rs)&\subseteq&\satmap_\QSAS(\qscs)\setminus\satmap_\QSAS(\rs')
\\
\es&\neq&\satmap_\QSAS(\rs')&\subseteq&\satmap_\QSAS(\qscs)\setminus\satmap_\QSAS(\rs)\;,
\end{array}
\] 
and so Eq.\eqref{eq-jddj} holds. 
\end{proof}

\section{Concluding remarks}

In this paper, we extended the general approach 
elaborated in~\cite{DBLP:journals/tcs/JanickiKKM21,DBLP:series/sci/2022-1020} 
to allow 
dealing with the semantics of concurrent systems with, \eg nested transactions. 
The specification of such behaviours 
was provided by relationships between events, namely \emph{precedence}
and \emph{weak precedence}, intuitively corresponding to the `earlier than' and 
`not later than' positions in individual executions. 
 
As far as applications are concerned,
the resulting framework has the potential of, \eg 
alleviating the state space explosion problem,
which is one of the most challenging problems in the 
verification of concurrent systems or optimal
(re)scheduling of transaction-based systems. 
 
Another direction of the planned work is related to the theory of 
traces~\cite{dr95}.
Traces that can be interpreted as sets of interval orders and 
represent concurrent histories --- 
called \emph{interval traces} --- were 
proposed and discussed in~\cite{DBLP:journals/iandc/JanickiY17}. 
Though~\cite{DBLP:journals/iandc/JanickiY17}
analysed their relationship with interval 
relational structures from~\cite{JanKou93,JanKou97}, 
the relationship with the \textsc{qsa}-structures introduced in this paper is yet to be developed. 
 
\bibliography{intervals}
 
\end{document}

%% file: macros.tex
\pgfdeclarelayer{background}
\pgfdeclarelayer{foreground}
\pgfsetlayers{background,main,foreground}

\newcommand{\DiagramBig}[1][0.5]
     {
     \begin{tikzpicture}[node distance=1.3cm,>=arrow30,%
     line  width=0.3mm,scale=#1,bend angle=45]
     \tikzstyle{box}=[draw,regular polygon,thick,%
     regular polygon sides=4,minimum size=22mm, inner sep = -3pt]
     \tikzstyle{ybox}=[draw,regular polygon,thick,%
     regular polygon sides=4,minimum size=22mm, inner sep = -3pt,fill=yellow]
     \tikzstyle{rbox}=[draw,regular polygon,thick,%
     regular polygon sides=4,minimum size=22mm, inner sep = -3pt,fill=red]
     \tikzstyle{gbox}=[draw,regular polygon,thick,%
     regular polygon sides=4,minimum size=22mm, inner sep = -3pt,fill=green]

     \tikzstyle{dybox}=[draw,regular polygon,thick,style=dashed,%
     regular polygon sides=4,minimum size=22mm, inner sep = -3pt,fill=yellow]
     \tikzstyle{drbox}=[draw,regular polygon,thick,style=dashed,%
     regular polygon sides=4,minimum size=22mm, inner sep = -3pt,fill=red]
     \tikzstyle{dgbox}=[draw,regular polygon,thick,style=dashed,%
     regular polygon sides=4,minimum size=22mm, inner sep = -3pt,fill=green]
     }

\newcommand{\Diagram}[1][0.5]
     {
     \begin{tikzpicture}[node distance=1.3cm,>=arrow30,%
     line  width=0.3mm,scale=#1,bend angle=45]
     \tikzstyle{box}=[draw,regular polygon,thick,rounded corners,%
     regular polygon sides=4,minimum size=14mm, inner sep = -3pt]
     \tikzstyle{ybox}=[draw,regular polygon,thick,rounded corners,%
     regular polygon sides=4,minimum size=14mm, inner sep = -3pt,fill=yellow!40]
     \tikzstyle{rbox}=[draw,regular polygon,thick,rounded corners,%
     regular polygon sides=4,minimum size=14mm, inner sep = -3pt,fill=red!40]
     \tikzstyle{gbox}=[draw,regular polygon,thick,rounded corners,%
     regular polygon sides=4,minimum size=14mm, inner sep = -3pt,fill=green!40]
     \tikzstyle{bbox}=[draw,regular polygon,thick,rounded corners,%
     regular polygon sides=4,minimum size=14mm, inner sep = -3pt,fill=cyan!40]

     \tikzstyle{dybox}=[draw,regular polygon,thick,style=dashed,rounded corners,%
     regular polygon sides=4,minimum size=14mm, inner sep = -3pt,fill=yellow!40]
     \tikzstyle{drbox}=[draw,regular polygon,thick,style=dashed,rounded corners,%
     regular polygon sides=4,minimum size=14mm, inner sep = -3pt,fill=red!40]
     \tikzstyle{dgbox}=[draw,regular polygon,thick,style=dashed,rounded corners,%
     regular polygon sides=4,minimum size=14mm, inner sep = -3pt,fill=green!40]
     \tikzstyle{dbbox}=[draw,regular polygon,thick,style=dashed,rounded corners,%
     regular polygon sides=4,minimum size=14mm, inner sep = -3pt,fill=cyan!40]
     }

\newcommand{\Relation}[2]
     {
     \begin{tikzpicture}[baseline=-0.5cm, node distance=.8cm,>=arrow30,line width=0.3mm]
     \tikzstyle{tlcorner}=[xshift=#2,yshift=#1]
     \tikzstyle{brcorner}=[xshift=#1,yshift=#2]
     \tikzstyle{bgcolor}=[fill=brown!20]
     }

\newcommand{\StandardNet}[1][0.5]
     {
     \begin{tikzpicture}[node distance=1.3cm,>=arrow30,line  width=0.3mm,scale=#1,auto,bend angle=45]
     \tikzstyle{place}=[draw,circle,thick,minimum size=4.5mm]
     \tikzstyle{transition}=
                [draw,regular polygon,thick,
                 regular polygon sides=4,minimum size=8mm, inner sep = -2pt]
     }

\newcommand{\StandardTS}[1][0.5]
     {
     \begin{tikzpicture}[node distance=0.5cm,>=arrow30,bend angle=45,scale=#1]
     }

\newcommand{\colorarc}[4]
    {\path [color=#1, -stealth, line width=#2,
     postaction={draw, line width=#3, shorten >=#4, -}]}

\pgfarrowsdeclare{arrow30}{arrow30}
{
  \arrowsize=0.2pt
  \advance\arrowsize by.275\pgflinewidth%
  \pgfarrowsleftextend{+-\arrowsize}
  \advance\arrowsize by.5\pgflinewidth
  \pgfarrowsrightextend{+\arrowsize}
}
{
  \arrowsize=0.2pt
  \advance\arrowsize by.275\pgflinewidth%
  \pgfsetdash{}{+0pt}
  \pgfsetroundjoin
  \pgfpathmoveto{\pgfqpoint{1\arrowsize}{0\arrowsize}}
  \pgfpathlineto{\pgfqpoint{-12\arrowsize}{2.5\arrowsize}}
  \pgfpathlineto{\pgfqpoint{-12\arrowsize}{-2.5\arrowsize}}
  \pgfpathclose
  \pgfusepathqfillstroke
}

\newcommand{\Ext}[2]         {\operatorname{ext}_{#1}(#2)}

\DeclareMathOperator{\cls}            {clo}

\DeclareMathOperator{\satmap}         {max}

\newcommand{\ie}           {i.e., }

\newcommand{\eg}           {e.g., }

\newcommand{\EOD}          {\hspace{5mm}$\diamond$}

\newcommand{\sat}            {\mathit{max}}

\newcommand{\closed}         {\mathit{clo}}

\newcommand{\sq}             {\sqsubset}
\newcommand{\es}             {\varnothing}

\newcommand{\FA}[1]          {\langle #1\rangle}

\newcommand{\pref}           {\trianglelefteq}

\newcommand{\po}             {\mathit{po}}
\newcommand{\rs}             {\mathit{rs}}
\newcommand{\rss}            {\mathit{RS}}

\newcommand*{\SigmaStar}[1][]
{\mathbin{\tikz [baseline=-0.25ex]
\draw [very thin]
(0pt,0.5ex) --
(0.25em,0.5ex) --
(0.125em,0.1ex) --
(0.25em,0.5ex) --
(0.125em,0.9ex) --
(0.25em,0.5ex) --
(0.35em,0.5ex) --
(0.225em,0.2ex) --
(0.4em,0.2ex) --
(0.225em,0.2ex) --
(0.35em,0.5ex) --
(0.225em,0.8ex)  --
(0.4em,0.8ex) --
(0.225em,0.8ex) --
(0.35em,0.5ex);
}}

\newcommand*{\SigmaAsterisc}[1][]
{\mathbin{\tikz [baseline=-0.25ex]
\draw [very thin]
(0pt,0.5ex) --
(0.25em,0.5ex) --
(0.125em,0.1ex) --
(0.375em,0.1ex) --
(0.125em,0.1ex) --
(0.25em,0.5ex) --
(0.125em,0.9ex) --
(0.375em,0.9ex) --
(0.125em,0.9ex) --
(0.25em,0.5ex) --
(0.3125em, 0.3ex) --
(0.25em,0.5ex) --
(0.3125em, 0.7ex) --
(0.25em,0.5ex) --
(0.5em,0.5ex);
}}

\newcommand{\RR}             {\mathcal{R}}

\newcommand{\SSS}            {\mathcal{S}}

%% file: figures/orders.tex
\newsavebox\total
 \sbox{\total}{%
\begin{tikzpicture}[remember picture, node distance=1.5cm,>=arrow30]
 \node (n1) {$\bullet$};
 \node (n2) [right of=n1] {$\bullet$};
 \node (n3) [right of=n2] {$\bullet$};
 \node (n4) [right of=n3] {$\bullet$};
 \node (f1) [below of=n3, yshift=.8cm] {};
 \node (f2) [above of=n3, yshift=-.8cm] {};

 \draw (n1) [->, thick] edge (n2);
 \draw (n2) [->, thick] edge (n3);
 \draw (n3) [->, thick] edge (n4);
 \draw (n1) [->, thick, out=25, in=155] edge (n3);
 \draw (n2) [->, thick, out=25, in=155] edge (n4);
 \draw (n1) [->, thick, out=-25, in=-155] edge (n4);

\begin{pgfonlayer}{background}
 \node (p0) [draw,fill=gray!10,fit=(n1)(n4)(f1)(f2),inner sep=1.1mm,color=gray!10] {};
\end{pgfonlayer}
\end{tikzpicture}%
}

\newsavebox\stratified
 \sbox{\stratified}{%
\begin{tikzpicture}[remember picture, node distance=1.5cm,>=arrow30]
 \node (n1) {$\bullet$};
 \node (n2) [right of=n1] {$\bullet$};
 \node (n3) [right of=n2] {$\bullet$};
 \node (n4) [right of=n3] {$\bullet$};
 \node (f1) [below of=n3, yshift=.8cm] {};
 \node (f2) [above of=n3, yshift=-.8cm] {};

 \draw (n2) [->, thick] edge (n3);
 \draw (n3) [->, thick] edge (n4);
 \draw (n1) [->, thick, out=25, in=155] edge (n3);
 \draw (n2) [->, thick, out=25, in=155] edge (n4);
 \draw (n1) [->, thick, out=-25, in=-155] edge (n4);

\begin{pgfonlayer}{background}
 \node (p0) [draw,fill=gray!10,fit=(n1)(n4)(f1)(f2),inner sep=1.1mm,color=gray!10] {};
\end{pgfonlayer}
\end{tikzpicture}%
}

\newsavebox\interval
 \sbox{\interval}{%
\begin{tikzpicture}[remember picture, node distance=1.5cm,>=arrow30]
 \node (n1) {$\bullet$};
 \node (n2) [right of=n1] {$\bullet$};
 \node (n3) [right of=n2] {$\bullet$};
 \node (n4) [right of=n3] {$\bullet$};
 \node (f1) [below of=n3, yshift=.8cm] {};
 \node (f2) [above of=n3, yshift=-.8cm] {};

 \draw (n1) [->, thick, out=25, in=155] edge (n3);
 \draw (n2) [->, thick, out=25, in=155] edge (n4);
 \draw (n1) [->, thick, out=-25, in=-155] edge (n4);

\begin{pgfonlayer}{background}
 \node (p0) [draw,fill=gray!10,fit=(n1)(n4)(f1)(f2),inner sep=1.1mm,color=gray!10] {};
\end{pgfonlayer}
\end{tikzpicture}%
}

\newsavebox\Mypartial
 \sbox{\Mypartial}{%
\begin{tikzpicture}[remember picture, node distance=1.5cm,>=arrow30]
 \node (n1) {$\bullet$};
 \node (n2) [right of=n1] {$\bullet$};
 \node (n3) [right of=n2] {$\bullet$};
 \node (n4) [right of=n3] {$\bullet$};
 \node (f1) [below of=n3, yshift=.8cm] {};
 \node (f2) [above of=n3, yshift=-.8cm] {};

 \draw (n1) [->, thick, out=25, in=155] edge (n3);
 \draw (n2) [->, thick, out=25, in=155] edge (n4);

\begin{pgfonlayer}{background}
 \node (p0) [draw,fill=gray!10,fit=(n1)(n4)(f1)(f2),inner sep=1.1mm,color=gray!10] {};
\end{pgfonlayer}
\end{tikzpicture}%
}

In Figure~\ref{fig-orders} one can see ($a$) total order, ($b$) stratified order which is not total, 
($c$) interval order which is not stratified, and ($d$) partial order which is not interval.
\EOD

\begin{figure}[ht]
\begin{center}
\begin{tikzpicture}[node distance=6.5cm]
 \node (A) {($a$) \usebox{\total}};
 \node (B) [right of=A, yshift=0mm] {($b$) \usebox{\stratified}};
 \node (C) [below of=A, yshift=45mm] {($c$) \usebox{\interval}};
 \node (D) [right of=C, yshift=0mm] {($d$) \usebox{\Mypartial}};
 \end{tikzpicture}

 \end{center}
 \caption
{
 The hierarchy of partial orders represented as directed graphs.
}
\label{fig-orders}
 \end{figure}
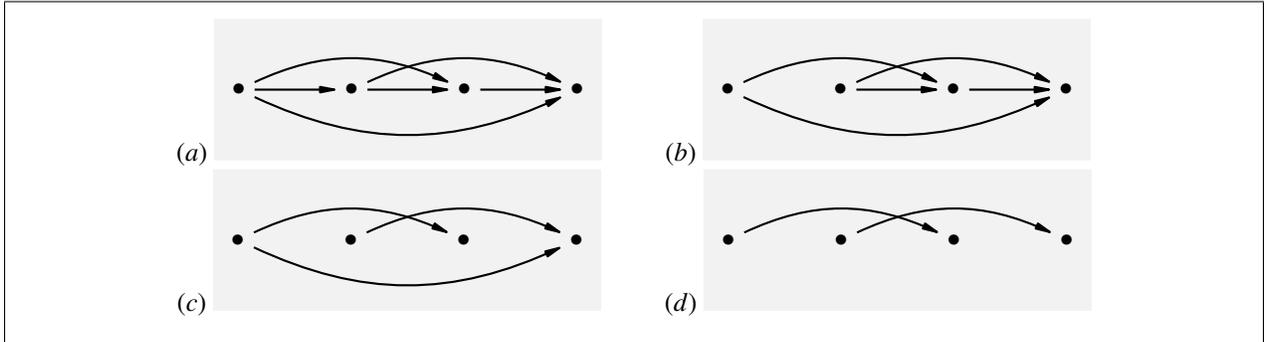

%% file: figures/fcycles.tex
\newsavebox\wcycle
 \sbox{\wcycle}{%
\begin{tikzpicture}[remember picture, node distance=1.5cm,>=arrow30]
 \node (n1) {$\bullet$};
 \node (n2) [right of=n1] {$\bullet$};
 \node (n3) [below of=n2] {$\bullet$};
 \node (n4) [left of=n3] {$\bullet$};
  
 \draw (n1) [->, dashed] edge (n2);
 \draw (n2) [->, dashed] edge (n3);
 \draw (n3) [->, dashed] edge (n4);
 \draw (n4) [->, dashed] edge (n1);

\begin{pgfonlayer}{background}
 \node (p0) [draw,fill=gray!10,fit=(n1)(n3),inner sep=1.1mm,color=gray!10] {};
\end{pgfonlayer}
\end{tikzpicture}%
}

\newsavebox\onecycle
 \sbox{\onecycle}{%
\begin{tikzpicture}[remember picture, node distance=1.5cm,>=arrow30]
 \node (n1) {$\bullet$};
 \node (n2) [right of=n1] {$\bullet$};
 \node (n3) [below of=n2] {$\bullet$};
 \node (n4) [left of=n3] {$\bullet$};
  
 \draw (n1) [->] edge (n2);
 \draw (n2) [->, dashed] edge (n3);
 \draw (n3) [->, dashed] edge (n4);
 \draw (n4) [->, dashed] edge (n1);

\begin{pgfonlayer}{background}
 \node (p0) [draw,fill=gray!10,fit=(n1)(n3),inner sep=1.1mm,color=gray!10] {};
\end{pgfonlayer}
\end{tikzpicture}%
}

\newsavebox\twocycle
 \sbox{\twocycle}{%
\begin{tikzpicture}[remember picture, node distance=1.5cm,>=arrow30]
 \node (n1) {$\bullet$};
 \node (n2) [right of=n1] {$\bullet$};
 \node (n3) [below of=n2] {$\bullet$};
 \node (n4) [left of=n3] {$\bullet$};
  
 \draw (n1) [->] edge (n2);
 \draw (n2) [->] edge (n3);
 \draw (n3) [->, dashed] edge (n4);
 \draw (n4) [->, dashed] edge (n1);

\begin{pgfonlayer}{background}
 \node (p0) [draw,fill=gray!10,fit=(n1)(n3),inner sep=1.1mm,color=gray!10] {};
\end{pgfonlayer}
\end{tikzpicture}%
}

\newsavebox\intcycle
 \sbox{\intcycle}{%
\begin{tikzpicture}[remember picture, node distance=1.5cm,>=arrow30]
 \node (n1) {$\bullet$};
 \node (n2) [right of=n1] {$\bullet$};
 \node (n3) [below of=n2] {$\bullet$};
 \node (n4) [left of=n3] {$\bullet$};
  
 \draw (n1) [->] edge (n2);
 \draw (n2) [->, dashed] edge (n3);
 \draw (n3) [->] edge (n4);
 \draw (n4) [->, dashed] edge (n1);

\begin{pgfonlayer}{background}
 \node (p0) [draw,fill=gray!10,fit=(n1)(n3),inner sep=1.1mm,color=gray!10] {};
\end{pgfonlayer}
\end{tikzpicture}%
}

In Figure~\ref{fig-fcycles} one can see ($a$) cycle forbidden for sequential semantic, but allowed in all other semantics, ($b$) and ($c$) cycles forbidden for sequential and step sequence semantics, but allowed in interval semantic,
($d$) cycle forbidden in all considered semantics.
\EOD

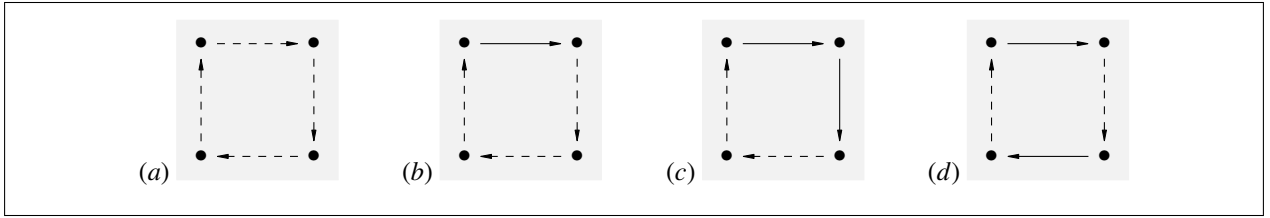
\begin{figure}[ht]
\begin{center}
\begin{tikzpicture}[node distance=3.5cm]
 \node (A) {($a$) \usebox{\wcycle}};
 \node (B) [right of=A, yshift=0mm] {($b$) \usebox{\onecycle}};
 \node (C) [right of=B, yshift=0mm] {($c$) \usebox{\twocycle}};
 \node (D) [right of=C, yshift=0mm] {($d$) \usebox{\intcycle}};
 \end{tikzpicture}

 \end{center}
 \caption
{
 The hierarchy of forbidden cycles.
}
\label{fig-fcycles}
 \end{figure}

%% file: figures/qsorder.tex
\newsavebox\quasi
 \sbox{\quasi}{%
\begin{tikzpicture}[remember picture, node distance=1.5cm,>=arrow30]
 \node (n1) {$\bullet$};
 \node (n2) [right of=n1] {$\bullet$};
 \node (n3) [right of=n2] {$\bullet$};
 \node (n4) [right of=n3] {$\bullet$};
 \node (f1) [below of=n3, yshift=.8cm] {};
 \node (f2) [above of=n3, yshift=-.8cm] {};

 \draw (n3) [->, thick] edge (n4);
 \draw (n1) [->, thick, out=25, in=155] edge (n3);
 \draw (n2) [->, thick, out=25, in=155] edge (n4);
 \draw (n1) [->, thick, out=-25, in=-155] edge (n4);

\begin{pgfonlayer}{background}
 \node (p0) [draw,fill=gray!10,fit=(n1)(n4)(f1)(f2),inner sep=1.1mm,color=gray!10] {};
\end{pgfonlayer}
\end{tikzpicture}%
}

\newsavebox\notquasi
 \sbox{\notquasi}{%
\begin{tikzpicture}[remember picture, node distance=1.5cm,>=arrow30]
 \node (n1) {$\bullet$};
 \node (n2) [right of=n1] {$\bullet$};
 \node (n3) [right of=n2] {$\bullet$};
 \node (n4) [right of=n3] {$\bullet$};
 \node (f1) [below of=n3, yshift=.8cm] {};
 \node (f2) [above of=n3, yshift=-.8cm] {};

 \draw (n1) [->, thick, out=25, in=155] edge (n3);
 \draw (n2) [->, thick, out=25, in=155] edge (n4);
 \draw (n1) [->, thick, out=-25, in=-155] edge (n4);

\begin{pgfonlayer}{background}
 \node (p0) [draw,fill=gray!10,fit=(n1)(n4)(f1)(f2),inner sep=1.1mm,color=gray!10] {};
\end{pgfonlayer}
\end{tikzpicture}%
}
In Figure~\ref{fig-quasi} one can see the completion of the hierarchy from the Figure~\ref{fig-orders} (the order in part $(b)$ is repeated as an example of interval order which is not \textsc{qs}-order). The order in part $(a)$ is \textsc{qs}-order since one can construct it by composing sequentially first and third element, then adding in parallel the second element and, finally, sequentially composing the result with the forth element (see Definition~\ref{def-qso}). 

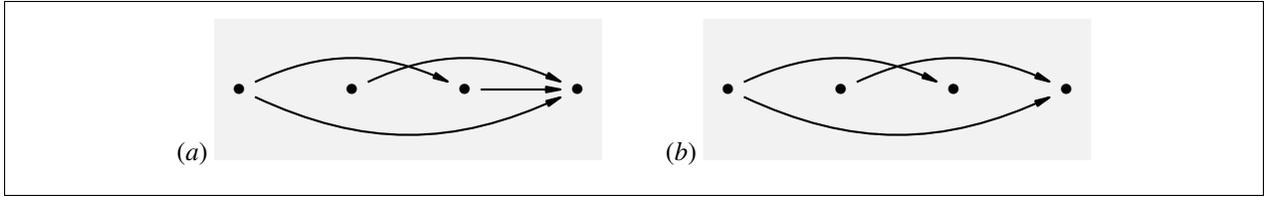
\begin{figure}[ht]
\begin{center} 
\begin{tikzpicture}[node distance=6.5cm]
 \node (A) {($a$) \usebox{\quasi}};
  \node (A) [right of=A] {($b$) \usebox{\notquasi}};
 \end{tikzpicture}

 \end{center}
 \caption
{
 A \textsc{qs}-order which is not stratified, and stratified order which is not \textsc{qs}-order.
}
\label{fig-quasi}
 \end{figure}

%% file: figures/max.tex
\newsavebox\quasistr
 \sbox{\quasistr}{%
\begin{tikzpicture}[remember picture, node distance=1.5cm,>=arrow30]
 \node (n1) {$a$};
 \node (n4) [right of=n1] {$d$};
 \node (n2) [below of=n1, yshift=8mm, xshift=7.5mm] {$c$};
 \node (n3) [above of=n1, yshift=-8mm,xshift=7.5mm] {$b$};

 \draw (n1) [->, thick] edge (n2);
 \draw (n3) [->, thick] edge (n4);
 \draw (n1) [->, dashed] edge (n3);
 \draw (n2) [->, dashed] edge (n4);

\begin{pgfonlayer}{background}
 \node (p0) [draw,fill=gray!10,fit=(n1)(n2)(n3)(n4),inner sep=1.1mm,color=gray!10] {};
\end{pgfonlayer}
\end{tikzpicture}%
}

\newsavebox\maxstr
 \sbox{\maxstr}{%
\begin{tikzpicture}[remember picture, node distance=1.5cm,>=arrow30]
 \node (n1) {$a$};
 \node (n4) [right of=n1] {$d$};
 \node (n2) [below of=n1, yshift=8mm, xshift=7.5mm] {$c$};
 \node (n3) [above of=n1, yshift=-8mm,xshift=7.5mm] {$b$};

 \draw (n1) [->, thick] edge (n2);
 \draw (n1) [->, thick] edge (n4);
 \draw (n3) [->, thick] edge (n4);
 \draw (n2) [->, thick] edge (n4);
 \draw (n1) [->, dashed, out=60, in=-150] edge (n3);
 \draw (n3) [->, dashed, out=-120, in=30] edge (n1);
 \draw (n2) [->, dashed, out=110, in=-110] edge (n3);
 \draw (n3) [->, dashed, out=-70, in=70] edge (n2);

\begin{pgfonlayer}{background}
 \node (p0) [draw,fill=gray!10,fit=(n1)(n2)(n3)(n4),inner sep=1.1mm,color=gray!10] {};
\end{pgfonlayer}
\end{tikzpicture}%
}

\newsavebox\intervalrepr
  \sbox{\intervalrepr}{%
\begin{tikzpicture}[remember picture, node distance=1.4cm,>=arrow30]

    \node (ab) {};
	\node (bb) [above of=ab, yshift=-4mm] {};
	\node (be) [right of=bb, xshift=-5mm] {};
	\node (cb) [right of=be, xshift=-15mm] {};
	\node (ce) [right of=cb, xshift=-5mm] {};
	\node (ae) [below of=ce, yshift=4mm] {};
    \node (db) [right of=ce, xshift=-15mm, yshift=-5mm] {};
	\node (de) [right of=db, xshift=-5mm] {};    
	\node (f1) [above of=bb, yshift=-9mm] {};
	
    \draw (ab) [thick] edge node (a) [above] {$b$} (ae);
	\draw (bb) [thick] edge node (b) [above] {$a$} (be);
	\draw (cb) [thick] edge node (c) [above] {$c$} (ce);
	\draw (db) [thick] edge node (d) [above] {$d$} (de);

\begin{pgfonlayer}{background}
	\node (p1) [draw,fill=gray!10,fit=(ae)(f1)(de),inner sep=1.1mm,color=gray!10] {};
\end{pgfonlayer}
\end{tikzpicture}%
}

Figure~\ref{fig-max} shows a \textsc{qsa}-structure together with one possible
\textsc{qsm}-structure extending it (on the right of the diagram, a 
corresponding interval representation is provided; cf.\ \cite{JKKM-PN2024}). 
The \textsc{qsm}-structure in the middle diagram corresponds (through the mapping $\rho$ introduced in
Section~\ref{subs-2}) to the  
\textsc{qso}-order depicted in Figure~\ref{fig-quasi} (with elements labelled $a\ldots d$ from left to right).

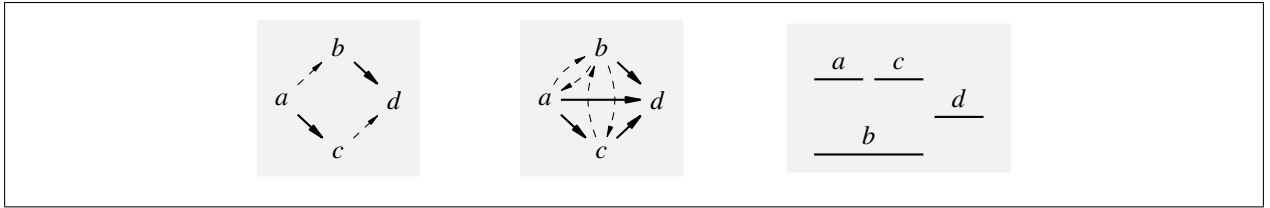
\begin{figure}[ht]
\begin{center}
\begin{tikzpicture}[node distance=3.5cm]
 \node (A) {\usebox{\quasistr}};
 \node (B) [right of=A] {\usebox{\maxstr}};
 \node (C) [right of=B, xshift = 4.5mm] {\usebox{\intervalrepr}};
 \end{tikzpicture}

 \end{center}
 \caption
{
 A \textsc{qsa}-structure together with one of its \textsc{qsm}-structures 
 and an interval representation of the corresponding \textsc{qs}-order. Solid and dashed arcs represent relations $\prec$ and $\sq$, respectively. 
 We take advantage of the fact that in the \textsc{qsm}-structures $\prec$ is included in $\sq$.
}
\label{fig-max}
 \end{figure}

%% file: figures/closed.tex
\newsavebox\quasistrX
 \sbox{\quasistrX}{%
\begin{tikzpicture}[remember picture, node distance=1.5cm,>=arrow30]
 \node (n1) {$a$};
 \node (n4) [right of=n1] {$d$};
 \node (n2) [below of=n1, yshift=8mm, xshift=7.5mm] {$c$};
 \node (n3) [above of=n1, yshift=-8mm,xshift=7.5mm] {$b$};

 \draw (n1) [->, thick] edge (n2);
 \draw (n3) [->, thick] edge (n4);
 \draw (n1) [->, dashed] edge (n3);
 \draw (n2) [->, dashed] edge (n4);

\begin{pgfonlayer}{background}
 \node (p0) [draw,fill=gray!10,fit=(n1)(n2)(n3)(n4),inner sep=1.1mm,color=gray!10] {};
\end{pgfonlayer}
\end{tikzpicture}%
}

\newsavebox\closedstr
 \sbox{\closedstr}{%
\begin{tikzpicture}[remember picture, node distance=1.5cm,>=arrow30]
 \node (n1) {$a$};
 \node (n4) [right of=n1] {$d$};
 \node (n2) [below of=n1, yshift=8mm, xshift=7.5mm] {$c$};
 \node (n3) [above of=n1, yshift=-8mm,xshift=7.5mm] {$b$};

 \draw (n1) [->, thick] edge (n2);
 \draw (n1) [->, thick] edge (n4);
 \draw (n3) [->, thick] edge (n4);
 \draw (n2) [->, dashed] edge (n4);
 \draw (n1) [->, dashed] edge (n3);

\begin{pgfonlayer}{background}
 \node (p0) [draw,fill=gray!10,fit=(n1)(n2)(n3)(n4),inner sep=1.1mm,color=gray!10] {};
\end{pgfonlayer}
\end{tikzpicture}%
}

In Figure~\ref{fig-closed} one can see a quasi-stratified structure from Example~\ref{ex-max} together with its closure.

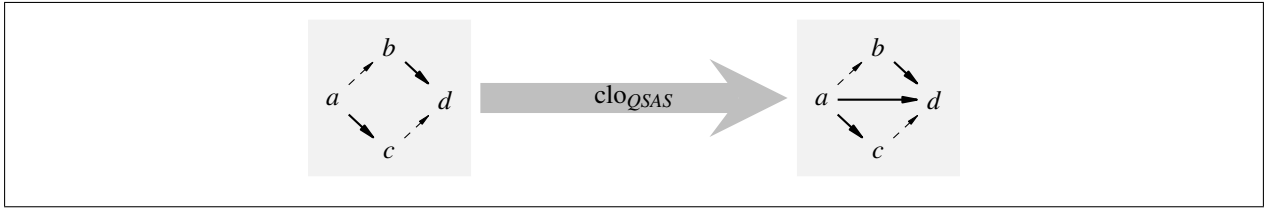
\begin{figure}[ht]
\begin{center}
\begin{tikzpicture}[node distance=6.5cm]
 \node (A) {\usebox{\quasistrX}};
 \node (B) [right of=A] {\usebox{\closedstr}};
\colorarc{lightgray}{4mm}{1mm}{7mm}(A) edge [-stealth] node [color=black] {$\cls_\QSAS$} (B);
 \end{tikzpicture}

 \end{center}
 \caption
{
 A \textsc{qsa}-structure and its closure. We take advantage of the fact that in \textsc{csc}-structures $\prec$ is included in $\sq$ (and only provide solid arcs for such relationships).
}
\label{fig-closed}
 \end{figure}

%% file: figures/allmax.tex
\newsavebox\quasistrXX
 \sbox{\quasistrXX}{%
\begin{tikzpicture}[remember picture, node distance=1.5cm,>=arrow30]
 \node (n1) {$a$};
 \node (n4) [right of=n1] {$d$};
 \node (n2) [below of=n1, yshift=8mm, xshift=7.5mm] {$c$};
 \node (n3) [above of=n1, yshift=-8mm,xshift=7.5mm] {$b$};

 \draw (n1) [->, thick] edge (n2);
 \draw (n3) [->, thick] edge (n4);
 \draw (n1) [->, dashed] edge (n3);
 \draw (n2) [->, dashed] edge (n4);

\begin{pgfonlayer}{background}
 \node (p0) [draw,fill=gray!10,fit=(n1)(n2)(n3)(n4),inner sep=0.6mm,color=gray!10] {};
\end{pgfonlayer}
\end{tikzpicture}%
}

\newsavebox\maxstrX
 \sbox{\maxstrX}{%
\begin{tikzpicture}[remember picture, node distance=1.5cm,>=arrow30]
 \node (n1) {$a$};
 \node (n4) [right of=n1] {$d$};
 \node (n2) [below of=n1, yshift=8mm, xshift=7.5mm] {$c$};
 \node (n3) [above of=n1, yshift=-8mm,xshift=7.5mm] {$b$};

 \draw (n1) [->, thick] edge (n2);
 \draw (n1) [->, thick] edge (n4);
 \draw (n3) [->, thick] edge (n4);
 \draw (n2) [->, thick] edge (n4);
 \draw (n1) [->, dashed, out=60, in=-150] edge (n3);
 \draw (n3) [->, dashed, out=-120, in=30] edge (n1);
 \draw (n2) [->, dashed, out=110, in=-110] edge (n3);
 \draw (n3) [->, dashed, out=-70, in=70] edge (n2);

\begin{pgfonlayer}{background}
 \node (p0) [draw,fill=gray!10,fit=(n1)(n2)(n3)(n4),inner sep=0.6mm,color=gray!10] {};
\end{pgfonlayer}
\end{tikzpicture}%
}

\newsavebox\intervalreprX
  \sbox{\intervalreprX}{%
\begin{tikzpicture}[remember picture, node distance=1.4cm,>=arrow30]

    \node (ab) {};
	\node (bb) [above of=ab, yshift=-4mm] {};
	\node (be) [right of=bb, xshift=-6mm] {};
	\node (cb) [right of=be, xshift=-12.5mm] {};
	\node (ce) [right of=cb, xshift=-6mm] {};
	\node (ae) [below of=ce, yshift=4mm] {};
    \node (db) [right of=ce, xshift=-12.5mm, yshift=-5mm] {};
	\node (de) [right of=db, xshift=-6mm] {};    
	\node (f1) [above of=bb, yshift=-9mm] {};
	
    \draw (ab) [thick] edge node (a) [above] {$b$} (ae);
	\draw (bb) [thick] edge node (b) [above] {$a$} (be);
	\draw (cb) [thick] edge node (c) [above] {$c$} (ce);
	\draw (db) [thick] edge node (d) [above] {$d$} (de);

\begin{pgfonlayer}{background}
	\node (p1) [draw,fill=gray!10,fit=(ae)(f1)(de),inner sep=1.1mm,color=gray!10] {};
\end{pgfonlayer}
\end{tikzpicture}%
}

\newsavebox\treereprX
  \sbox{\treereprX}{%
    \begin{tikzpicture}[node distance=1.3cm,>=arrow30,line  width=0.3mm,scale=1.1,auto,bend angle=45,font=\scriptsize,inner sep=1pt]

     \node(a)at(1,1) [label=above:$a$]{\normalsize $\bullet$};
     \node(b)at(1.3,0) [label=below:$b$]{\normalsize $\bullet$};
     \node(c)at(1.6,1) [label=above:$c$]{\normalsize $\bullet$};
     \node(d)at(1.9,0) [label=below:$d$]{\normalsize $\bullet$};

     \path (a) edge (b);
     \path (c) edge (b);

\begin{pgfonlayer}{background}
	\node (p1) [draw,fill=gray!10,fit=(a)(d),inner sep=3mm,color=gray!10] {};
\end{pgfonlayer}
     
    \end{tikzpicture}
}

\newsavebox\maxstrXX
 \sbox{\maxstrXX}{%
\begin{tikzpicture}[remember picture, node distance=1.5cm,>=arrow30]
 \node (n1) {$a$};
 \node (n4) [right of=n1] {$d$};
 \node (n2) [below of=n1, yshift=8mm, xshift=7.5mm] {$c$};
 \node (n3) [above of=n1, yshift=-8mm,xshift=7.5mm] {$b$};

 \draw (n1) [->, thick] edge (n2);
 \draw (n1) [->, thick] edge (n4);
 \draw (n3) [->, thick] edge (n4);
 \draw (n1) [->, thick] edge (n3);
 \draw (n2) [->, dashed, out=60, in=-150] edge (n4);
 \draw (n4) [->, dashed, out=-120, in=30] edge (n2);
 \draw (n2) [->, dashed, out=110, in=-110] edge (n3);
 \draw (n3) [->, dashed, out=-70, in=70] edge (n2);

\begin{pgfonlayer}{background}
 \node (p0) [draw,fill=gray!10,fit=(n1)(n2)(n3)(n4),inner sep=0.6mm,color=gray!10] {};
\end{pgfonlayer}
\end{tikzpicture}%
}

\newsavebox\intervalreprXX
  \sbox{\intervalreprXX}{%
\begin{tikzpicture}[remember picture, node distance=1.4cm,>=arrow30]

    \node (ab) {};
	\node (bb) [above of=ab, yshift=-4mm] {};
	\node (be) [right of=bb, xshift=-6mm] {};
	\node (cb) [right of=be, xshift=-15mm] {};
	\node (ce) [right of=cb, xshift=-6mm] {};
	\node (ae) [below of=ce, yshift=4mm] {};
    \node (de) [left of=bb, xshift=15mm, yshift=-5mm] {};
	\node (db) [left of=de, xshift=6mm] {};    
	\node (f1) [above of=bb, yshift=-9mm] {};
	
    \draw (ab) [thick] edge node (a) [above] {$c$} (ae);
	\draw (bb) [thick] edge node (b) [above] {$b$} (be);
	\draw (cb) [thick] edge node (c) [above] {$d$} (ce);
	\draw (db) [thick] edge node (d) [above] {$a$} (de);

\begin{pgfonlayer}{background}
	\node (p1) [draw,fill=gray!10,fit=(ae)(f1)(db),inner sep=1.1mm,color=gray!10] {};
\end{pgfonlayer}
\end{tikzpicture}%
}

\newsavebox\treereprXX
  \sbox{\treereprXX}{%
    \begin{tikzpicture}[node distance=1.3cm,>=arrow30,line  width=0.3mm,scale=1.1,auto,bend angle=45,font=\scriptsize,inner sep=1pt]

     \node(a)at(1,1) [label=above:$b$]{\normalsize $\bullet$};
     \node(b)at(1.3,0) [label=below:$c$]{\normalsize $\bullet$};
     \node(c)at(1.6,1) [label=above:$d$]{\normalsize $\bullet$};
     \node(d)at(0.7,0) [label=below:$a$]{\normalsize $\bullet$};

     \path (a) edge (b);
     \path (c) edge (b);

\begin{pgfonlayer}{background}
	\node (p2) [draw,fill=gray!10,fit=(d)(c),inner sep=3mm,color=gray!10] {};
\end{pgfonlayer}
     
    \end{tikzpicture}
}

\newsavebox\maxstrXXX
 \sbox{\maxstrXXX}{%
\begin{tikzpicture}[remember picture, node distance=1.5cm,>=arrow30]
 \node (n1) {$a$};
 \node (n4) [right of=n1] {$d$};
 \node (n2) [below of=n1, yshift=8mm, xshift=7.5mm] {$c$};
 \node (n3) [above of=n1, yshift=-8mm,xshift=7.5mm] {$b$};

 \draw (n1) [->, thick] edge (n2);
 \draw (n1) [->, thick] edge (n4);
 \draw (n3) [->, thick] edge (n4);
 \draw (n2) [->, thick] edge (n4);
 \draw (n1) [->, dashed, out=60, in=-150] edge (n3);
 \draw (n3) [->, dashed, out=-120, in=30] edge (n1);
 \draw (n3) [->, thick] edge (n2);

\begin{pgfonlayer}{background}
 \node (p0) [draw,fill=gray!10,fit=(n1)(n2)(n3)(n4),inner sep=0.6mm,color=gray!10] {};
\end{pgfonlayer}
\end{tikzpicture}%
}

\newsavebox\intervalreprXXX
  \sbox{\intervalreprXXX}{%
\begin{tikzpicture}[remember picture, node distance=1.4cm,>=arrow30]

    \node (ab) {};
	\node (bb) [above of=ab, yshift=-4mm] {};
	\node (be) [right of=bb, xshift=-6mm] {};
	\node (cb) [right of=be, xshift=-12.5mm, yshift=-5mm] {};
	\node (ce) [right of=cb, xshift=-6mm] {};
	\node (ae) [below of=be, yshift=4mm] {};
    \node (db) [right of=ce, xshift=-12.5mm] {};
	\node (de) [right of=db, xshift=-6mm] {};    
	\node (f1) [above of=bb, yshift=-9mm] {};
	
    \draw (ab) [thick] edge node (a) [above] {$b$} (ae);
	\draw (bb) [thick] edge node (b) [above] {$a$} (be);
	\draw (cb) [thick] edge node (c) [above] {$c$} (ce);
	\draw (db) [thick] edge node (d) [above] {$d$} (de);

\begin{pgfonlayer}{background}
	\node (p1) [draw,fill=gray!10,fit=(ae)(f1)(de),inner sep=1.1mm,color=gray!10] {};
\end{pgfonlayer}
\end{tikzpicture}%
}

\newsavebox\treereprXXX
  \sbox{\treereprXXX}{%
    \begin{tikzpicture}[node distance=1.3cm,>=arrow30,line  width=0.3mm,scale=1.1,auto,bend angle=45,font=\scriptsize,inner sep=1pt]
     \node(f)at(1,1.1) {};
     \node(b)at(1,0) [label=below:$ab$]{\normalsize $\bullet$};
     \node(c)at(1.45,0) [label=below:$c$]{\normalsize $\bullet$};
     \node(d)at(1.9,0) [label=below:$d$]{\normalsize $\bullet$};


\begin{pgfonlayer}{background}
	\node (p1) [draw,fill=gray!10,fit=(b)(f)(d),inner sep=3mm,color=gray!10] {};
\end{pgfonlayer}
     
    \end{tikzpicture}
}

\newsavebox\maxstrXXXX
 \sbox{\maxstrXXXX}{%
\begin{tikzpicture}[remember picture, node distance=1.5cm,>=arrow30]
 \node (n1) {$a$};
 \node (n4) [right of=n1] {$d$};
 \node (n2) [below of=n1, yshift=8mm, xshift=7.5mm] {$c$};
 \node (n3) [above of=n1, yshift=-8mm,xshift=7.5mm] {$b$};

 \draw (n1) [->, thick] edge (n2);
 \draw (n1) [->, thick] edge (n4);
 \draw (n3) [->, thick] edge (n4);
 \draw (n1) [->, thick] edge (n3);
 \draw (n2) [->, dashed, out=60, in=-150] edge (n4);
 \draw (n4) [->, dashed, out=-120, in=30] edge (n2);
 \draw (n3) [->, thick] edge (n2);

\begin{pgfonlayer}{background}
 \node (p0) [draw,fill=gray!10,fit=(n1)(n2)(n3)(n4),inner sep=0.6mm,color=gray!10] {};
\end{pgfonlayer}
\end{tikzpicture}%
}

\newsavebox\intervalreprXXXX
  \sbox{\intervalreprXXXX}{%
\begin{tikzpicture}[remember picture, node distance=1.4cm,>=arrow30]

    \node (ab) {};
	\node (bb) [above of=ab, yshift=-4mm] {};
	\node (be) [right of=bb, xshift=-6mm] {};
	\node (ce) [left of=bb, xshift=15mm, yshift=-5mm] {};
	\node (cb) [left of=ce, xshift=6mm] {};
	\node (ae) [below of=be, yshift=4mm] {};
	\node (de) [left of=cb, xshift=15mm] {};    
    \node (db) [left of=de, xshift=6mm] {};
	\node (f1) [above of=bb, yshift=-9mm] {};
	
    \draw (ab) [thick] edge node (a) [above] {$c$} (ae);
	\draw (bb) [thick] edge node (b) [above] {$d$} (be);
	\draw (cb) [thick] edge node (c) [above] {$b$} (ce);
	\draw (db) [thick] edge node (d) [above] {$a$} (de);

\begin{pgfonlayer}{background}
	\node (p1) [draw,fill=gray!10,fit=(ae)(f1)(db),inner sep=1.1mm,color=gray!10] {};
\end{pgfonlayer}
\end{tikzpicture}%
}

\newsavebox\treereprXXXX
  \sbox{\treereprXXXX}{%
    \begin{tikzpicture}[node distance=1.3cm,>=arrow30,line  width=0.3mm,scale=1.1,auto,bend angle=45,font=\scriptsize,inner sep=1pt]

     \node(f)at(1,1.1) {};
     \node(b)at(1,0) [label=below:$a$]{\normalsize $\bullet$};
     \node(c)at(1.45,0) [label=below:$b$]{\normalsize $\bullet$};
     \node(d)at(1.9,0) [label=below:$cd$]{\normalsize $\bullet$};


\begin{pgfonlayer}{background}
	\node (p1) [draw,fill=gray!10,fit=(c)(b)(d)(f),inner sep=3mm,color=gray!10] {};
\end{pgfonlayer}
     
    \end{tikzpicture}
}

\newsavebox\maxstrV
 \sbox{\maxstrV}{%
\begin{tikzpicture}[remember picture, node distance=1.5cm,>=arrow30]
 \node (n1) {$a$};
 \node (n4) [right of=n1] {$d$};
 \node (n2) [below of=n1, yshift=8mm, xshift=7.5mm] {$c$};
 \node (n3) [above of=n1, yshift=-8mm,xshift=7.5mm] {$b$};

 \draw (n1) [->, thick] edge (n2);
 \draw (n1) [->, thick] edge (n4);
 \draw (n3) [->, thick] edge (n4);
 \draw (n1) [->, dashed, out=60, in=-150] edge (n3);
 \draw (n3) [->, dashed, out=-120, in=30] edge (n1);
 \draw (n2) [->, dashed, out=60, in=-150] edge (n4);
 \draw (n4) [->, dashed, out=-120, in=30] edge (n2);
 \draw (n3) [->, thick] edge (n2);

\begin{pgfonlayer}{background}
 \node (p0) [draw,fill=gray!10,fit=(n1)(n2)(n3)(n4),inner sep=0.6mm,color=gray!10] {};
\end{pgfonlayer}
\end{tikzpicture}%
}

\newsavebox\intervalreprV
  \sbox{\intervalreprV}{%
\begin{tikzpicture}[remember picture, node distance=1.4cm,>=arrow30]

    \node (ab) {};
	\node (bb) [above of=ab, yshift=-4mm] {};
	\node (be) [right of=bb, xshift=-6mm] {};
	\node (ce) [left of=bb, xshift=15mm] {};
	\node (cb) [left of=ce, xshift=6mm] {};
	\node (ae) [below of=be, yshift=4mm] {};
	\node (de) [below of=ce, yshift=4mm] {};    
    \node (db) [below of=cb, yshift=4mm] {};
	\node (f1) [above of=bb, yshift=-9mm] {};
	
    \draw (ab) [thick] edge node (a) [above] {$c$} (ae);
	\draw (bb) [thick] edge node (b) [above] {$d$} (be);
	\draw (cb) [thick] edge node (c) [above] {$b$} (ce);
	\draw (db) [thick] edge node (d) [above] {$a$} (de);

\begin{pgfonlayer}{background}
	\node (p1) [draw,fill=gray!10,fit=(ae)(f1)(db),inner sep=1.1mm,color=gray!10] {};
\end{pgfonlayer}
\end{tikzpicture}%
}

\newsavebox\treereprV
  \sbox{\treereprV}{%
    \begin{tikzpicture}[node distance=1.3cm,>=arrow30,line  width=0.3mm,scale=1.1,auto,bend angle=45,font=\scriptsize,inner sep=1pt]

     \node(f)at(1,1.1) {};
     \node(b)at(1,0) [label=below:$ab$]{\normalsize $\bullet$};
     \node(d)at(1.45,0) [label=below:$cd$]{\normalsize $\bullet$};


\begin{pgfonlayer}{background}
	\node (p1) [draw,fill=gray!10,fit=(d)(b)(f),inner sep=3mm,color=gray!10] {};
\end{pgfonlayer}
     
    \end{tikzpicture}
}

\newsavebox\maxstrVX
 \sbox{\maxstrVX}{%
\begin{tikzpicture}[remember picture, node distance=1.5cm,>=arrow30]
 \node (n1) {$a$};
 \node (n4) [right of=n1] {$d$};
 \node (n2) [below of=n1, yshift=8mm, xshift=7.5mm] {$c$};
 \node (n3) [above of=n1, yshift=-8mm,xshift=7.5mm] {$b$};

 \draw (n1) [->, thick] edge (n2);
 \draw (n1) [->, thick] edge (n4);
 \draw (n3) [->, thick] edge (n4);
 \draw (n1) [->, thick] edge (n3);
 \draw (n2) [->, thick] edge (n4);
 \draw (n3) [->, thick] edge (n2);

\begin{pgfonlayer}{background}
 \node (p0) [draw,fill=gray!10,fit=(n1)(n2)(n3)(n4),inner sep=0.6mm,color=gray!10] {};
\end{pgfonlayer}
\end{tikzpicture}%
}

\newsavebox\intervalreprVX
  \sbox{\intervalreprVX}{%
\begin{tikzpicture}[remember picture, node distance=1.4cm,>=arrow30]

    \node (abX) {};
    \node (ab) [above of=abX, yshift=-9mm] {};
	\node (ae) [right of=ab, xshift=-6mm] {};
	\node (bb) [right of=ae, xshift=-15mm] {};
	\node (be) [right of=bb, xshift=-6mm] {};
	\node (cb) [right of=be, xshift=-15mm] {};
	\node (ce) [right of=cb, xshift=-6mm] {};
	\node (db) [right of=ce, xshift=-15mm] {};    
    \node (de) [right of=db, xshift=-6mm] {};
	\node (f1) [above of=bb, yshift=-4mm] {};
	
    \draw (ab) [thick] edge node (a) [above] {$a$} (ae);
	\draw (bb) [thick] edge node (b) [above] {$b$} (be);
	\draw (cb) [thick] edge node (c) [above] {$c$} (ce);
	\draw (db) [thick] edge node (d) [above] {$d$} (de);

\begin{pgfonlayer}{background}
	\node (p1) [draw,fill=gray!10,fit=(abX)(ae)(f1)(de),inner sep=1.1mm,color=gray!10] {};
\end{pgfonlayer}
\end{tikzpicture}%
}

\newsavebox\treereprVX
  \sbox{\treereprVX}{%
    \begin{tikzpicture}[node distance=1.3cm,>=arrow30,line  width=0.3mm,scale=1.1,auto,bend angle=45,font=\scriptsize,inner sep=1pt]

     \node(f)at(1,1.1) {};
     \node(a)at(1,0) [label=below:$a$]{\normalsize $\bullet$};
     \node(b)at(1.3,0) [label=below:$b$]{\normalsize $\bullet$};
     \node(c)at(1.6,0) [label=below:$c$]{\normalsize $\bullet$};
     \node(d)at(1.9,0) [label=below:$d$]{\normalsize $\bullet$};

\begin{pgfonlayer}{background}
	\node (p1) [draw,fill=gray!10,fit=(f)(a)(d),inner sep=3mm,color=gray!10] {};
\end{pgfonlayer}
     
    \end{tikzpicture}
}

\newsavebox\maxstrVXX
 \sbox{\maxstrVXX}{%
\begin{tikzpicture}[remember picture, node distance=1.5cm,>=arrow30]
 \node (n1) {$a$};
 \node (n4) [right of=n1] {$d$};
 \node (n2) [below of=n1, yshift=8mm, xshift=7.5mm] {$c$};
 \node (n3) [above of=n1, yshift=-8mm,xshift=7.5mm] {$b$};

 \draw (n1) [->, thick] edge (n2);
 \draw (n1) [->, thick] edge (n4);
 \draw (n3) [->, thick] edge (n4);
 \draw (n1) [->, thick] edge (n3);
 \draw (n2) [->, thick] edge (n4);
 \draw (n2) [->, thick] edge (n3);

\begin{pgfonlayer}{background}
 \node (p0) [draw,fill=gray!10,fit=(n1)(n2)(n3)(n4),inner sep=0.6mm,color=gray!10] {};
\end{pgfonlayer}
\end{tikzpicture}%
}

\newsavebox\intervalreprVXX
  \sbox{\intervalreprVXX}{%
\begin{tikzpicture}[remember picture, node distance=1.4cm,>=arrow30]

    \node (abX) {};
    \node (ab) [above of=abX, yshift=-9mm] {};
	\node (ae) [right of=ab, xshift=-6mm] {};
	\node (bb) [right of=ae, xshift=-15mm] {};
	\node (be) [right of=bb, xshift=-6mm] {};
	\node (cb) [right of=be, xshift=-15mm] {};
	\node (ce) [right of=cb, xshift=-6mm] {};
	\node (db) [right of=ce, xshift=-15mm] {};    
    \node (de) [right of=db, xshift=-6mm] {};
	\node (f1) [above of=bb, yshift=-4mm] {};
	
    \draw (ab) [thick] edge node (a) [above] {$a$} (ae);
	\draw (bb) [thick] edge node (b) [above] {$c$} (be);
	\draw (cb) [thick] edge node (c) [above] {$b$} (ce);
	\draw (db) [thick] edge node (d) [above] {$d$} (de);

\begin{pgfonlayer}{background}
	\node (p1) [draw,fill=gray!10,fit=(abX)(ae)(f1)(de),inner sep=1.1mm,color=gray!10] {};
\end{pgfonlayer}
\end{tikzpicture}%
}

\newsavebox\treereprVXX
  \sbox{\treereprVXX}{%
    \begin{tikzpicture}[node distance=1.3cm,>=arrow30,line  width=0.3mm,scale=1.1,auto,bend angle=45,font=\scriptsize,inner sep=1pt]

     \node(f)at(1,1.1) {};
     \node(a)at(1,0) [label=below:$a$]{\normalsize $\bullet$};
     \node(b)at(1.3,0) [label=below:$c$]{\normalsize $\bullet$};
     \node(c)at(1.6,0) [label=below:$b$]{\normalsize $\bullet$};
     \node(d)at(1.9,0) [label=below:$d$]{\normalsize $\bullet$};

\begin{pgfonlayer}{background}
	\node (p1) [draw,fill=gray!10,fit=(f)(a)(d),inner sep=3mm,color=gray!10] {};
\end{pgfonlayer}
     
    \end{tikzpicture}
}

\newsavebox\maxstrVXXX
 \sbox{\maxstrVXXX}{%
\begin{tikzpicture}[remember picture, node distance=1.5cm,>=arrow30]
 \node (n1) {$a$};
 \node (n4) [right of=n1] {$d$};
 \node (n2) [below of=n1, yshift=8mm, xshift=7.5mm] {$c$};
 \node (n3) [above of=n1, yshift=-8mm,xshift=7.5mm] {$b$};

 \draw (n1) [->, thick] edge (n2);
 \draw (n1) [->, thick] edge (n4);
 \draw (n3) [->, thick] edge (n4);
 \draw (n2) [->, thick] edge (n4);
 \draw (n1) [->, thick] edge (n3);
 \draw (n2) [->, dashed, out=110, in=-110] edge (n3);
 \draw (n3) [->, dashed, out=-70, in=70] edge (n2);

\begin{pgfonlayer}{background}
 \node (p0) [draw,fill=gray!10,fit=(n1)(n2)(n3)(n4),inner sep=0.6mm,color=gray!10] {};
\end{pgfonlayer}
\end{tikzpicture}%
}

\newsavebox\intervalreprVXXX
  \sbox{\intervalreprVXXX}{%
\begin{tikzpicture}[remember picture, node distance=1.4cm,>=arrow30]

    \node (ab) {};
	\node (bb) [above of=ab, yshift=-4mm] {};
	\node (be) [right of=bb, xshift=-6mm] {};
	\node (ae) [below of=be, yshift=4mm] {};
    \node (db) [right of=be, xshift=-15mm, yshift=-5mm] {};
	\node (de) [right of=db, xshift=-6mm] {};    
    \node (ce) [left of=bb, xshift=15mm, yshift=-5mm] {};
	\node (cb) [left of=ce, xshift=6mm] {};
	\node (f1) [above of=bb, yshift=-9mm] {};
	
    \draw (ab) [thick] edge node (a) [above] {$b$} (ae);
	\draw (bb) [thick] edge node (b) [above] {$c$} (be);
	\draw (cb) [thick] edge node (c) [above] {$a$} (ce);
	\draw (db) [thick] edge node (d) [above] {$d$} (de);

\begin{pgfonlayer}{background}
	\node (p1) [draw,fill=gray!10,fit=(de)(f1)(cb)(ab),inner sep=1.1mm,color=gray!10] {};
\end{pgfonlayer}
\end{tikzpicture}%
}

\newsavebox\treereprVXXX
  \sbox{\treereprVXXX}{%
    \begin{tikzpicture}[node distance=1.3cm,>=arrow30,line  width=0.3mm,scale=1.1,auto,bend angle=45,font=\scriptsize,inner sep=1pt]

     \node(f)at(1,1.1) {};
          \node(h)at(1.9,1.1) {};
     \node(a)at(1,0) [label=below:$a$]{\normalsize $\bullet$};
     \node(b)at(1.45,0) [label=below:$bc$]{\normalsize $\bullet$};
     \node(d)at(1.9,0) [label=below:$d$]{\normalsize $\bullet$};


\begin{pgfonlayer}{background}
	\node (p1) [draw,fill=gray!10,fit=(a)(d)(f)(h),inner sep=3mm,color=gray!10] {};
\end{pgfonlayer}
     
    \end{tikzpicture}
}


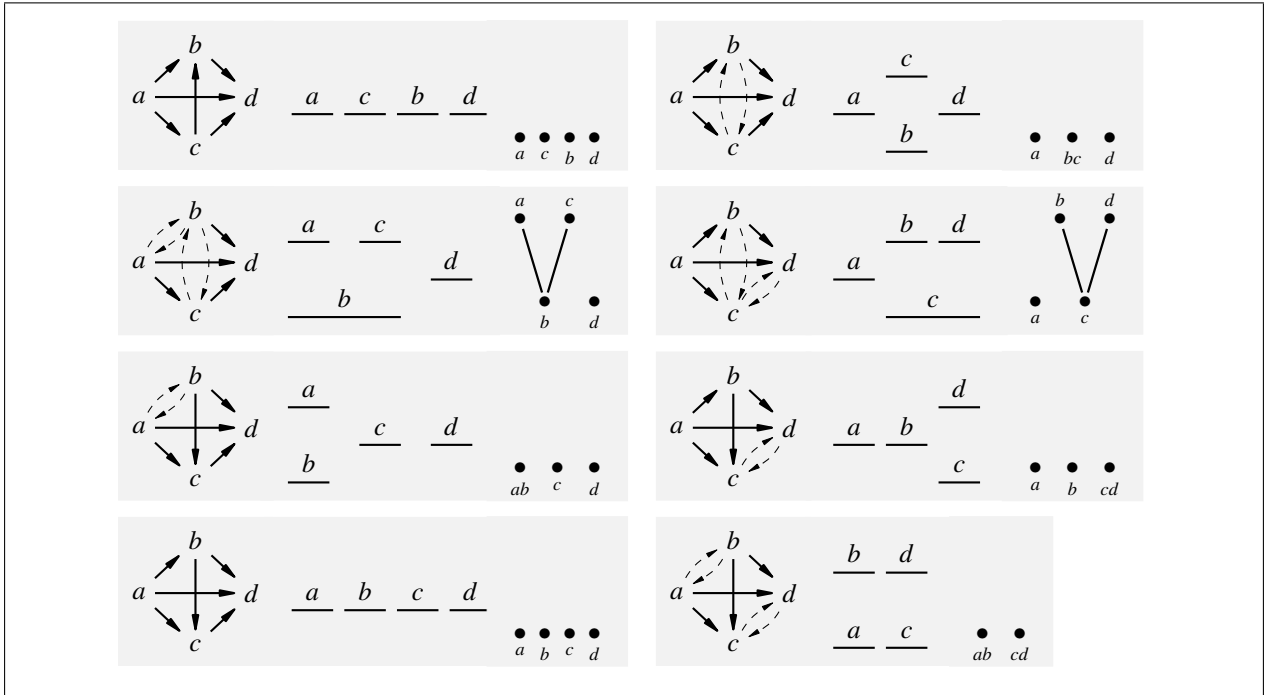
\begin{figure}[ht]
\begin{center}
\begin{tikzpicture}[node distance=2.2cm]
 \node (H1) {\usebox{\maxstrVXX}};
 \node (H2) [right of=H1, xshift = 4mm] {\usebox{\intervalreprVXX}};
 \node (H3) [right of=H2, xshift = .5mm] {\usebox{\treereprVXX}}; 
 \node (I1) [right of=H3, xshift = 1mm] {\usebox{\maxstrVXXX}};
 \node (I2) [right of=I1, xshift = 1mm] {\usebox{\intervalreprVXXX}};
 \node (I3) [right of=I2, xshift = .5mm] {\usebox{\treereprVXXX}};
 \node (B1) [below of=H1] {\usebox{\maxstrX}};
 \node (B2) [right of=B1, xshift = 2.5mm] {\usebox{\intervalreprX}};
 \node (B3) [right of=B2, xshift = 2mm] {\usebox{\treereprX}};
 \node (C1) [right of=B3, xshift = 1mm] {\usebox{\maxstrXX}};
 \node (C2) [right of=C1, xshift = 1mm] {\usebox{\intervalreprXX}};
 \node (C3) [right of=C2, xshift = .5mm] {\usebox{\treereprXX}};
 \node (D1) [below of=B1] {\usebox{\maxstrXXX}};
 \node (D2) [right of=D1, xshift = 2.5mm] {\usebox{\intervalreprXXX}};
 \node (D3) [right of=D2, xshift = 2mm] {\usebox{\treereprXXX}};
 \node (E1) [right of=D3, xshift = 1mm] {\usebox{\maxstrXXXX}};
 \node (E2) [right of=E1, xshift = 1mm] {\usebox{\intervalreprXXXX}};
 \node (E3) [right of=E2, xshift = .5mm] {\usebox{\treereprXXXX}}; 
 \node (G1) [below of=D1] {\usebox{\maxstrVX}};
 \node (G2) [right of=G1, xshift = 4mm] {\usebox{\intervalreprVX}};
 \node (G3) [right of=G2, xshift = .5mm] {\usebox{\treereprVX}}; 
 \node (F1) [right of=G3, xshift = 1mm] {\usebox{\maxstrV}};
 \node (F2) [right of=F1, xshift = -2.5mm] {\usebox{\intervalreprV}};
 \node (F3) [right of=F2, xshift = -5.5mm] {\usebox{\treereprV}};
 \end{tikzpicture}
 \end{center}
 \caption
{
 \textsc{qsm}-structures together with their interval and \textsc{qs}-sequence representations.
 Note that the nodes of the trees are are in fact sets,\eg  $a$ stands for the singleton
 set $\{a\}$, and $bc$ stands for the  set $\{b,c\}$. 
}
\label{fig-allmax}
 \end{figure}